\documentclass[11pt]{article}

\usepackage[utf8]{inputenc}

\usepackage[table,svgnames]{xcolor}
\usepackage[colorlinks=true,linkcolor=black,citecolor=MidnightBlue,urlcolor=MidnightBlue,unicode]{hyperref}
\usepackage{fullpage}

\usepackage{microtype}
\usepackage{mathtools,amssymb,mathrsfs,amsthm}   
\usepackage{thmtools,thm-restate}
\usepackage{placeins}

\usepackage{xspace}
\usepackage{tikz}

\usepackage{comment}
\usepackage{framed}
\usepackage{url}
\usepackage{algorithm}
\usepackage{algorithmicx}
\usepackage{algpseudocode}
\usepackage[labelfont=bf,font=small]{caption}
\usepackage{subcaption}
\usepackage[rightcaption]{sidecap}
\usepackage{tcolorbox}
\usetikzlibrary {shapes,patterns,patterns.meta}

\DeclareMathAlphabet{\mathbbold}{U}{bbold}{m}{n}

\newtheorem{theorem}{Theorem}[section]
\newtheorem{lemma}[theorem]{Lemma}

\newtheorem{corollary}[theorem]{Corollary}

\newtheorem{definition}{Definition}[section]

\newcommand{\val}{\operatorname{\mathit{val}}}
\newcommand{\F}{\mathbb{F}}
\newcommand{\N}{\mathbb{N}}
\newcommand{\Z}{\mathbb{Z}}
\newcommand{\eps}{\varepsilon}

\newcommand{\OPT}{\textup{\textsf{OPT}}}

\newcommand{\LB}{\gamma}

\newcommand{\OUT}{\mathsf{OUT}}
\newcommand{\RANK}{\mathsf{R}}
\newcommand{\Left}{\mathsf{left}}
\newcommand{\Right}{\mathsf{right}}
\newcommand{\COST}{\mathsf{COST}}
\newcommand{\END}{\mathsf{END}}
\newcommand{\Root}{\text{\textit{root}\/}}  %

\newcommand\E[1]{\mathbb{E}\left[\,#1\,\right]}
\renewcommand{\Pr}[1]{\mathbb{P}\left[\,#1\,\right]}
\newcommand{\setO}[1]{[#1]}

\renewcommand{\P}{\mathbb{P}}
\renewcommand{\L}{\mathbb{L}}

\newif\ifdraft
\draftfalse

\title{\textbf{Sorting in One and Two Rounds using $t$-Comparators}}

\author{Ran Gelles\thanks{Bar-Ilan University, Israel. \texttt{ran.gelles@biu.ac.il}} \and Zvi Lotker\thanks{Bar-Ilan University, Israel. \texttt{zvi.lotker@biu.ac.il}} \and Frederik Mallmann-Trenn\thanks{King's College London, United Kingdom. \texttt{frederik.mallmann-trenn@kcl.ac.ukl}}}  
\date{}

\begin{document}

\maketitle

\begin{abstract}
We examine sorting algorithms for $n$ elements whose basic operation is comparing $t$ elements simultaneously (a $t$-comparator). 
We focus on algorithms that use only a single round or two rounds---comparisons performed in the second round depend on the outcomes of the first round comparators. 

We design deterministic and randomized algorithms. 
In the deterministic case, we show an interesting relation to design theory (namely, to 2-Steiner systems), which yields a single-round optimal algorithm for $n=t^{2^k}$ with any $k\ge 1$ and a variety of possible values of~$t$. 
For some values of~$t$, however, no algorithm can reach the optimal (information-theoretic) bound on the number of comparators.  
For this case (and any other $n$ and $t$), we show an algorithm that uses at most three times as many comparators as the theoretical bound.

We also design a randomized Las-Vegas two-rounds sorting algorithm for any $n$ and~$t$. 
Our algorithm uses an asymptotically optimal number of $O(\max(\frac{n^{3/2}}{t^2},\frac{n}{t}))$ comparators, with high probability, i.e., with probability at least $1-1/n$. 
The analysis of this algorithm involves the gradual unveiling of randomness, using a novel technique which we coin the \emph{binary tree of deferred randomness}.
\end{abstract}

\section{Introduction}
\emph{Sorting} has been a fundamental task for computers (and earlier electronic devices) 
since the inception of computer history~\cite{knuthVol3-Sorting,CLRS}. Many sorting algorithms are \emph{comparison-based}, meaning that there exists some device that compares pairs of elements and decides which of them is the larger. By comparing multiple pairs, one can obtain a full order of all elements. It is well known that if pairs are being compared, $\Theta(n\log n)$ comparisons are needed in order to fully sort any possible set of $n$~elements.   Such sorting, however, assumes one can apply comparisons in an adaptive manner, i.e., one can determine which pairs of elements to compare next based on results of previous comparisons. It is not too difficult to see that without this adaptive selection of elements, $\Omega(n^2)$ comparisons are needed (see also Lemma~\ref{lem:t2allPairs} below).

In contrast to general-purpose CPUs, which allow fast comparison of two elements, specialized hardware that can be found in system-on-a-chip systems and GPUs, allows comparing larger sets of elements.
Motivated by the above, 
in this work we explore sorting algorithms that use $t$-comparators. These blocks allow $t$ elements to be compared simultaneously to determine their total order, rather than comparing them in pairs. 
Our initial focus is on deterministic, non-adaptive sorting algorithms where all comparisons are pre-determined and independent of prior outcomes. Additionally, we consider randomized algorithms with a limited degree of adaptiveness. In particular, we design sorting algorithms with two rounds,  where the second round can use the comparison outcomes from the first round. In both cases, our goal is to minimize the number of  $t$-comparators used.

To further motivate the case of sorting with $t$-comparators ($t$-sorting) in a single round, consider the following scenario, which is very common in the Computer Science community. A conference program committee (PC) is set to decide on the ranking of the $n$ submitted papers. Let's assume that there is an ``absolute truth'', namely, that there exists a total ordering of the papers, and that each PC member outputs the ``true'' ordering of any number of papers assigned to them.\footnote{We realize that, in real life, no such absolute truth exists, and that PC members are heavily biased, etc. These extensions make very interesting direction for followup questions. We briefly discuss future directions in Section~\ref{sec:conclusion}.}
To balance out the load, the papers are split  so that each PC member receives $t$~papers. Note that the same paper can be sent to multiple PC members. 
Each PC member, individually, returns to the chair the total order of the set of papers assigned to them. The chair collects all these outputs and composes a total ordering of the  $n$ papers, that is consistent with all the partial sets. Assume we wish the chair's output to be the ``true'' ordering of the papers, how many PC members are needed, as a function of $n$ and~$t$? Note that the chair assigns the papers once, without having any information about papers, that is, this is a non-adaptive $t$-sorting with a single round.

\subsection{Deterministic Sorting}
Consider deterministic $t$-sorting algorithms with a single round. Similar to the case of $t=2$, that requires comparing all $\binom{n}{2}$ possible pairs, it can easily be shown that for any $t$, at least $\LB_{n,t}=\binom{n}{2}/\binom{t}{2}$ 
many $t$-comparators are needed in order to fully sort $n$ elements. This stems from the fact that in order to learn the total ordering of $n$ elements, we need to learn the relative-order of all $\binom{n}{2}$ pairs, while each $t$-comparator gives us information about at most $\binom{t}{2}$ different pairs of elements (Corollary~\ref{cor:triv-lowerBound}). 

Our first question is whether this bound is achievable, that is, whether there exists a single round $t$-sorting algorithm that utilizes exactly $\LB_{n,t}$ comparators. 
We first show a way to perform $t$-sorting with at most $3 \LB_{n,t}$ comparators (Lemma~\ref{lem:loose-upperbound}). The idea is rather simple: we divide the elements into disjoint subsets, where each subset contains $t/2$ different elements. Then, we go through all possible pairs of subsets, and for each such pair we compare the $t$ elements of their union using a separate $t$-comparator. This guarantees that any two elements are compared by at least one comparator, so a total-ordering of the $n$ elements can be deduced from the results of the $\binom{ \lceil n/(t/2)\rceil}{2}<3\LB_{n,t}$ different comparisons.

Our main result is an algorithm with an optimal level of $\LB_{n,t}$ $t$-comparators for the case where $t$ is a power of a prime and $n=t^{2^k}$, for any positive integer $k\in\N$. Namely, 
\begin{restatable}[main, deterministic]{theorem}{RESTATABLEmainDeterministic}
\label{thm:main}
Let $t$ be a power of a prime and let $n=t^{2^k}$, $k\in\N$. Then, there exists a  deterministic single-round, $t$-sorting algorithm  that utilizes exactly $\binom{n}{2}/\binom{t}{2}$ comparators.
\end{restatable}
In order to obtain the above optimal sorting, we show a connection between sorting and combinatorial design theory. 
Consider the case where  $t$ is a prime power and $k=1$, that is, $n=t^2$, a setting that  attained a lot of interest in the past, especially by hardware-implementation oriented designs~\cite{TK77,SSS86,SS86}.
We essentially show that sorting with $\LB_{n,t}$ comparators is equivalent to an \emph{Affine Plane} of order~$t$. An affine plane (see e.g., \cite{Hughes_Piper_1985,Pascoe18})
is a design structure composed of elements (``points'') and subset of elements (``lines'') that guarantees the following properties: (P1) every two points belong to a \emph{unique} line, (P2) every line contains at least two points, and (P3) not all points are co-linear. Further, it satisfies the Euclidean Property  (A1): for every line~$L$ and  any point~$p$ outside~$L$, there exists  a unique line that contains~$p$ and is parallel to~$L$. 
It is known that all lines in an affine plane  contain exactly  the same number of points; call this number \emph{the order of the plane}. It is also known that an Affine plane of order $t$ contains $t^2+t$ lines.

If we think about points as the element we wish to sort and about lines as subsets of $t$ points which we compare via a single comparator, finding an affine plane of order $t$ provides the property that any two elements are being compared \emph{exactly} a single time, i.e., by a single comparator, leading to the optimal bound of~$\LB_{n,t}$.

An affine plane of order $t$ is easy to construct for any $t$ that is a power of a prime. Let $\F$ be a finite field with $t$ elements, and consider pairs of elements $(x,y)$, i.e., the plane $\F^2$. 
In this plane, any two points $(x_1,y_1)$ and $(x_2,y_2)$, define a unique line that passes through them, namely $y=\frac{y_1-y_2}{x_1-x_2}x+\frac{y_2x_1-y_1x_2}{x_1-x_2}$ if $x_1\ne x_2$ and the line $\{ (x_1,y) \mid y\in \F\}$, otherwise. 
It is easy to verify this structure satisfies all the properties of an affine plane (see \cite[Section~3.2]{Pascoe18}).

Affine planes are a special case of a more general combinatorial structures known as Steiner systems (Definition~\ref{def:steiner}). Indeed, if we change assumption (A1) so that there exists no parallel lines at all (also known as the Elliptic Property), but still require that any two points define a unique line, we would still get a sorting algorithm in which any two elements are being compared against each other exactly one. 
In this case, the resulting structure is again a special case of a Steiner system known as a \emph{Projective Plane}. 
Known constructions of projective planes imply that for any $t-1$ being a power of a prime, one can sort $t^2-t+1$ elements using exactly $t^2-t+1$ many $t$-comparators, where every pair of elements is being compared exactly once. 
These two constructions are summarized as  Theorem~\ref{thm:minimal-sorting-params}.

We lift the above result from optimally sorting $t^2$ elements to optimally sorting $t^{2^k}$, 
by developing 
a composition theorem (Lemma~\ref{lem:composition}) that recursively performs  sorting of $t^{2^k}$ elements by utilizing an optimal number of  $t^{2^{k-1}}$-comparators, for any~$k>1$.

\subsection{Randomized Sorting}
\label{sec:intro-rand}
Similar to the deterministic case, if one does not bound the number of adaptive rounds the algorithm is allowed to make, 
optimal sorting can be achieved. 
For instance, Beigel and Gill~\cite{beigel1990sorting}
showed a generalized $t$-\emph{quicksort} algorithm that sorts $n$ elements by utilizing at most $4\frac{n\log n}{t\log t}$ many $t$-comparators, which is optimal, maybe up to the constant (see Theorem~\ref{thm:OPT-rand}). However, this algorithm requires $O(\log_t n)$ adaptive rounds. Indeed, recall that quicksort works in rounds, where at each round the algorithm selects (one or more) pivot elements. These elements are used to ``bucket'' the rest of the elements into disjoint subsets, meaning that all  elements greater than one pivot and less than the next pivot belong to the same bucket. Then, each such bucket is recursively sorted by the same method.
Since each round depends on the pivots and buckets of the previous rounds, $O(\log_t n)$ recursive rounds are needed~\cite{beigel1990sorting}.

Our second question in this work is how to obtain optimal randomized $t$-sorting algorithms with restricted number of rounds.
Since we already analyzed the  case of a single round and reached optimal results, in the second part of this work we address the case of \emph{two} rounds. Our goal is to minimize the number of $t$-comparators used to sort $n$ elements in a Las-Vegas algorithm, where the output is correct with probability~1 but the number  \emph{of comparators} is a random variable that varies between different instances. 

Our main result for this part is as follows.
\begin{restatable}[main, randomized]{theorem}{RESTATABLEmainRandomized}
\label{thm:main-rand-informal}
Let $t<n$ be given. 
There exists a (Las-Vegas) randomized sorting algorithm for $n$ elements with two rounds, that utilizes $O\left(\max\left(\frac{n^{3/2}}{t^2},\frac{n}{t}\right)\right)$ many $t$-comparators, with probability at least $1-1/n$.
\end{restatable}
We note that for the case where $n=t^2$, our algorithm uses $O(t)$ comparators which is asymptotically optimal  since $\frac{n\log n}{t\log t}=\Theta(t)$. 
We further note that a result by Alon and Azar~\cite{AA87} implies that the expected number of comparators used in our algorithm, when $t\le\sqrt{n}$, is also tight.

The high-level idea of the two-round algorithm is to perform a single round of ``quicksort'' and then to optimally (deterministically) sort each resulting bucket, rather than recursively sorting it.
In more details, let $m$ be some fixed parameter. 
Our algorithm starts by sampling $m$ elements that will serve as pivots. We bucket all the elements by dividing the rest $n-m$ elements into subsets of size $m$ elements each, and comparing each such subset, along with the $m$ pivots by utilizing at most $3\LB_{2m,t}$ many $t$-comparators (per subset). This step tells us, for each one of the $n-m$ elements, between which two pivots it resides. %

A pseudo code of our 2-round randomized algorithm is given below as Algorithm~\ref{alg:rand-general-intro} for the case $t\leq m$. The case $t > m$ is very similar and is covered in Section~\ref{sec:random}.
\begin{algorithm}[ht]
\caption{A randomized 2-round sorting for any $n,t$ with $t \leq m$}
\label{alg:rand-general-intro}
\begin{algorithmic}[1]
\Statex
\Statex \textit{\textbf{Round 1:}}
\State Let $P$ be a set of $m$
elements from $A$, each sampled uniformly and independently from~$A$.%
\State Partition $A \setminus P$ into subsets $A_1,\ldots, A_{k}$ of size at most $m$ each. \Comment{$k=\lceil (n-m)/m\rceil$}
\ForAll {$i\in[k]$}
    \State Sort $P\cup A_i$ using the optimal 1-round deterministic algorithm.
\EndFor

\Statex
\Statex \textit{\textbf{Round 2:}}
\State 
Let $P=(p_1, \ldots, p_{m})$ be the ordered elements in~$P$. For $1\le i \le m-1$, set $S_i$ to contain all the elements which are greater than $p_i$ but lower than $p_{i+1}$. Set $S_0$ to be all the elements lower than $p_1$ and $S_{m}$ be all the elements greater than $p_{m}$.
\ForAll {$0 \le i \le m $}
    \State Sort $S_i$ using the optimal 1-round deterministic algorithm.
\EndFor
\end{algorithmic}
\end{algorithm}

In expectation, each bucket is of size $\approx n/m$ and sorting a bucket of this size takes  $3\LB_{n/m,t}$ many $t$-comparators. 
If all buckets had size exactly~$n/m$, this would lead immediately to the desired result of 
$3\frac{n}{m}\LB_{2m,t}+ 3(m+1)\LB_{n/m,t}= O(\frac{nm}{t^2}+\frac{n^2}{mt^2})$. 
This quantity is minimal when $m\approx \sqrt{n}$ (ignoring constants), leading to the claimed~$O(\frac{n^{3/2}}{t^2})$.\footnote{The term $O(n/t)$ in Theorem~\ref{thm:main-rand-informal} stems from the other case,  where $t>m$, i.e., $t>\sqrt{n}$.}
Unfortunately, buckets' sizes vary, and some of them might be much larger, say, of size $(n/m) \log(n/m)$. However, our analysis shows that this event is very rare and the additional number of comparators needed to handle these cases is rather small. More specifically, in our analysis, we formulate a balls-into-bins process to distribute elements into buckets, and bound the number of such bad events using the balls-into-bins process. Let us now expend on the techniques used in this analysis.

\subsubsection{Techniques: The binary tree of deferred randomness}
\label{sec:highlevel}
Let us start by describing the balls-into-bins process we use. 
Consider the $n$ elements, and rename them $a_1,\ldots, a_n$ 
so that they are sorted. Starting with $a_1$, we group together
sequences of~$c n/m$ consecutive elements, for some sufficiently large constant~$c$.
We call each such group \emph{a bin}; namely,
the first bin is $b_1=\{a_1,\ldots,a_{c n/m}\}$ the second bin is
$b_2 =\{ a_{c n/m +1}, \ldots, a_{2c n/m}\}$ and so on, resulting in a total of $m/c$~bins overall.
The balls will be the $m$ elements we pick as pivots. 
That is, let $P=\{p_1, p_2, \dots, p_m\}$ be the elements selected as pivots. 
Since each pivot is sampled uniformly at random, the selection of some $p_i$ is
equivalent to throwing a ball to bin~$b_j$ where $p_i\in b_j$.\footnote{We note that this balls-into-bin process differs slightly from our pivot selection process in the sense that it samples pivots with replacement, while the original process samples without replacement.
However, one could modify the original process by allowing the same element to be sampled multiple times, and later ignore these extra copies. 
It is immediate that sampling without replacement can only create smaller bins and thus improve the overall complexity.}

If each bin has a ball, than each ``bucket'' has at most $2cn/m$ elements, and the cost, measured in the number of comparators needed to sort that bucket, is as desired. 
However, the absence of a ball in a bin implies larger buckets. 
That is, the size of the bucket, and hence the cost of sorting it, 
is determined by the stretch of bins without balls (up to two additional bins, one from each side).
In other words, in order to bound the cost of the second round, 
we throw $|P|=m$ balls uniformly at random into $m/c$ bins 
and count the length of consecutive \emph{empty} bins.
Recall that $m=\sqrt{n}$; we will  substitute this value to avoid cumbersome equations in the following.

A straightforward balls-into-bins analysis shows that there are $c$ pivots per bin in expectation and that the probability of not having a pivot in $c'$ consecutive  bins scales as~$e^{-\Omega(c')}$. 
Ideally we would like to use the above probability and obtain a polynomially-small failure probability by considering all the bins at the same time.
Unfortunately, this approach breaks due to the correlation between empty bins. 
Indeed, the fact that some bins are empty indicates that the balls went somewhere else, altering the probability of having empty bins elsewhere.
The bins' loads are negatively correlated. This means that concentration bounds could potentially be used for negatively correlated variables. 
However, there are many obstacles to this approach.
First, note that while the loads of the bins are negatively correlated, we actually need to bound different variables, namely, the lengths of consecutive sequences of empty bin. 
Second, defining these variables and analyzing their probability function, as well as proving that they are negatively correlated, seems to be a difficult task.
Finally,  note that even the number of these random variables, is itself a random variable.

\begin{tcolorbox}[boxrule=1pt,leftrule=3mm]
One might think that the number of empty bins can be bounded by  analyzing the following process:
\textit{Form $m+1$ bins and throw balls  uniformly into them to determine the load of each bin (i.e., the size of each resulting bucket). Use the fact that the loads are negatively correlated to reach concentration.}
Unfortunately, this process does \emph{not} correspond to our algorithm.
To see this, consider the simple case of $n=10$ elements and a single pivot, yielding two buckets. 
When choosing a uniform pivot, the probability that all elements end up in the same bucket is 2/10, which happens either when the maximum element or the minimum element is the pivot. 
However, if we uniformly throw $9$~balls into two bins, they will all end up in the same bin with a probability of $2\cdot(1/2)^{9} = 1/256$. 
\end{tcolorbox}

Instead, we introduce the concept of a \emph{binary tree of deferred randomness}
which is instrumental in allowing us to do a simpler analysis of the concentration of empty bins while averting the obstacles caused by their dependencies. 
We think of the assignment of a pivot (a ball) to a bin 
as the bit-string describing the bin where the pivot ends, that we reveal bit-by-bit.
We define a binary tree, where each one of the $\sqrt{n}/c$ bins is a leaf. 
Thus, the tree has a depth of $\log(\sqrt{n}/c)$ (assuming $\sqrt{n}/c$ is a power of two).
We define the following iterative process of assigning balls to the leaves of the tree:
Initially we have $\sqrt{n}$ balls at the root. 
At every step, at every node~$u$, we randomly assign each ball to one of $u$'s children. This is equivalent to revealing the next bit in the string representing the bin to which the pivot belongs to.
The advantage of this approach lies in the careful revelation of the randomness. 
At every level, we can derive concentration
bounds without affecting the following levels---the only thing that matters at a given node is how many balls arrive at it.

Consider the binary tree of deferred randomness after all balls are assigned and follow an arbitrary path from the root to a leaf~$v$.
There are two cases.
In the ideal case, at every node along the path to~$v$, the number of balls going left and right is close to the expected value, namely, close to half.
If this happens, then enough balls propagate along this path and with high probability at least one of them will reach the leaf~$v$. This is the good scenario, since if this holds for many bins, the cost of sorting their elements will be very close to the expected cost.

The second case is when the concentration fails at some node~$u$ on the path, and the assignments of the balls is not close to half.
If this happens first at node $u$, we say that the bad event $\mathcal{E}_u$ occurred, stop the process there (i.e., ignore other nodes in $u$'s subtree), 
and charge a cost as if only a single ball reaches the bins under the node~$u$.
In other words, if there are 
$\ell$ balls at node~$u$, 
we assume that 
all the $\ell$ pivot selections ended up picking the same element.
By doing so, we overestimate the size of the resulted bucket to contain all the elements in all the bins below~$u$. 
Specifically, we charge this event with the cost of sorting $O(\ell \cdot c\sqrt{n}) > |\text{bins}(u)| \cdot c\sqrt{n}$ elements; here we use the fact that, as long as the bad event $\mathcal{E}_u$ does not happen, the number of balls reaching~$u$ always exceeds the number of bins in the subtree of~$u$, $|\text{bins}(u)|$. %

\begin{figure}
    \centering


\begin{tikzpicture}
  [scale=1,  
   GLabel/.style={lower node part/.style={green!55!black}},
   RLabel/.style={red!85!black,scale=0.8},
   vEND/.style={black!15,text=black,draw=black},
   every node/.style={fill=white,draw=black,circle,inner sep=2pt,minimum size=1.1cm},
   every lower node part/.style={font=\scriptsize,font=\bfseries},
   level 1/.style={level distance=14mm,sibling distance=48mm},
   level 2/.style={level distance=14mm,sibling distance=24mm,label distance=-2.5mm},
   level 3/.style={level distance=14mm,sibling distance=12mm,label distance=0mm,nodes={circle split}},
   level 4/.style={level distance=14mm,sibling distance=6mm,minimum height=10mm,nodes={minimum size=6mm,rectangle,draw=black}}]

  \node [circle split] {$u_1$ \nodepart{lower} \small 1600}
     child {node  [circle split] {$u_2$ \nodepart{lower} 790}
       child {node [circle split] {$u_4$ \nodepart{lower} 310}   
         child {node  {$u_{8}$ \nodepart{lower} 127}  
         	child {node [vEND] {$b_1$}}
		child {node [vEND] {$b_2$}}
         }
         child {node  {$u_9$\nodepart{lower} 183}	
         	child {node [vEND] {$b_3$}}
		child {node [vEND] {$b_4$}}
         }
       }
       child {node (u5) [circle split,vEND] {$u_5$ \nodepart{lower} 480}
         child {[densely dotted] node [circle split]  {$u_{10}$ \nodepart{lower} 470}	
                 child {node {$b_5$}}
		          child {node {$b_6$}}
         }
         child {[densely dotted] node [circle split]  {$u_{11}$  \nodepart{lower} 10}
         	  child {node {$b_7$}}
		      child {node (b8) {$b_8$}}         	
         }
       }
     }
     child {node [circle split] {$u_3$ \nodepart{lower} 810}
       child {node [circle split] {$u_6$ \nodepart{lower} 422}  
         child {[label distance=-2.5mm] node (u12) [vEND] {$u_{12}$ \nodepart{lower} 145}
         	  child {[densely dotted] node  (b9) {$b_9$}}
		      child {[densely dotted] node  {$b_{10}$}}         	
         }
         child {node {$u_{13}$ \nodepart{lower} 277}
         	  child {node [vEND] {$b_{11}$}}
		      child {node  [vEND]  {$b_{12}$}}     
	     }
       }
       child {node  [circle split] {$u_7$ \nodepart{lower} 388}  
       		child {node {$u_{14}$ \nodepart{lower} 194}
         		child {node [vEND] {$b_{13}$}}
			child {node [vEND] {$b_{14}$}}   		
		}
		child {node {$u_{15}$ \nodepart{lower} 194}
		        	child {node [vEND] {$b_{15}$}}
			child {node [vEND] {$b_{16}$}}   
		}
       }
     };
     
     \path (u5.north east) +(1mm,1mm) coordinate (e5);
     \node  at (e5)
     [circle,fill=red!60!black,text=white,inner sep=0pt,minimum size=0pt,scale=0.8]  {$\mathcal{E}_{u_5}$};
     \path (u12.north west)+(-0.5mm,1.5mm) coordinate (e12);
          \node  at (e12)
     [circle,fill=red!60!black,text=white,inner sep=-1pt,minimum size=0pt,scale=0.8]  {$\mathcal{E}_{u_{12}}$};
     

  \path (b8) + (0,-0.5cm) coordinate (BBox);
 
\end{tikzpicture}
    \caption{The figure shows the distribution of pivots (balls) on the \emph{tree of deferred randomness}, marked as the numbers in each node. 
Here we have $\sqrt{n}=1600$ pivots and 16 bins ($c=100$). 
In the first two levels, the distribution is about even. 
The node $u_{11}$ receives too few balls and so the event $\mathcal{E}_{u_5}$ holds. 
Similarly, $b_9$ gets too few balls (bin's balls are not shown in the figure), causing $\mathcal{E}_{u_{12}}$ to happen. 
The nodes in gray portray the set $\END$ described in Section~\ref{sec:random:analysis}, which includes the nodes where our deferred randomness revelation process stops. The cost of the scheme scales with the number of balls reaching nodes in this set.
}
    \label{fig:example}
\end{figure}
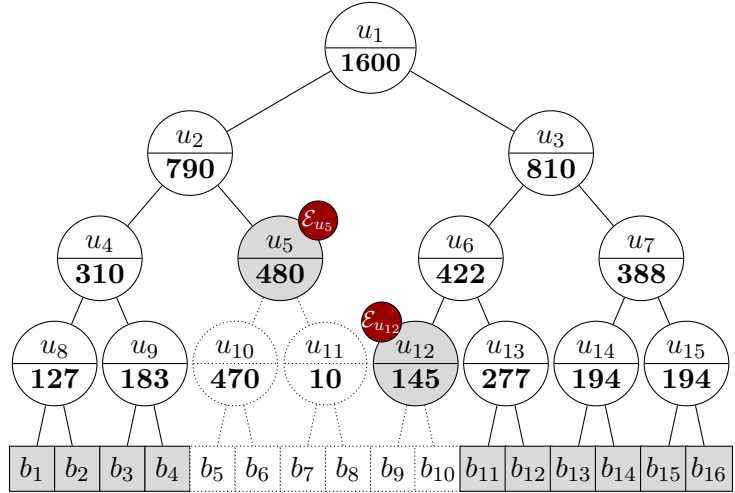

Figure~\ref{fig:example} illustrates the infiltration of balls through the tree: 
a node $u$ at level $i$ is associated to the $2^{\log(c \sqrt{n})-i}$ bins below it.
The number inside a node denotes how many balls are assigned to that node.
When $u$ assigns the balls to its children, each ball picks one of the children uniformly at random, so each of the children is assigned half of $u$'s balls, in expectation. 
The process continues until we reach the leaves at level $\log(c \sqrt{n})$. 
In the rare event that balls are distributed in a very skewed manner, the bad event $\mathcal{E}_u$~happens.
For instance, while $u_5$ has $480$~balls, they split very unevenly among its children, causing the bad event~$\mathcal{E}_{u_5}$. The process stops there, i.e., we do not care how the balls continue in the subtree of~$u_5$ and in particular,  $\mathcal{E}_u$ never happens in any of $u$'s descendants.  
Since $\mathcal{E}_{u_5}$ happens and the process stops there, the analysis charges an \emph{amortized} cost which is proportional to sorting a bucket of size of $4$ bins (due to the $4$ bins $b_5, b_6, b_7, b_8$---for all we know, all the balls could end up in~$b_8$, creating a single bucket that consists all the respective elements). 
In fact, we upper bound this cost by the number of balls that arrive to~$u_5$, whose expectation in this example is  $4c \gg 4$. 
The situation might get even worse, since $\mathcal{E}_{12}$ occurs as well. This effectively means that a single bucket might consist of all the element in bins $b_4$--$b_{10}$. The dependency between neighboring nodes with bad events complicates the cost analysis. However, by summing up the costs of all these events, we can derive the amortized cost per such bad event and simplify the analysis by considering a single event at a time.

Luckily, the higher up in the tree a node is, the more balls the node holds and the less likely the concentration bound will fail.
The lower in the tree the node is, the lower the cost is.
In particular, once we approach the lower levels of the tree, the bad event $\mathcal{E}_u$ occurs with constant probability. This does not pose any trouble, because the cost in this case is only a constant factor larger than the expected cost of the case where each 
bin has at least one ball in it.
Overall, we show that for every level of the tree, the cost imposed in our process is very close to its expectation, with high probability (at least $1-1/n^2$).
Taking a union bound over all the (at most $n$) levels of the our tree of deferred randomness yields the desired claim. We give the full details in Section~\ref{sec:rand:general}.

\subsection{Related Work}
A fundamental task like sorting naturally attracted a lot of attention in numerous variants and settings. To put our result in the right context, in this section we mention just a few of these variants and we mainly focus on \emph{comparison-based} sorting algorithms. We refer the reader to  surveys~\cite{10.1145/356593.356594,10.1145/146370.146381,singh2018survey} and books~\cite{knuthVol3-Sorting,CLRS,akl1985} for a more complete treatment on the background of (general) sorting.

The task of sorting in small number of rounds was initiated by the work of 
H\"{a}ggkvist and Hell~\cite{HH81}, who considered the case of sorting $n$ elements in a single round by comparing pairs of elements (i.e., $t=2$). While they do not give any explicit  sorting algorithm, they bound the number of $2$-comparators required for sorting in $d$-rounds by $\Omega(n^{1+1/d})$ from below and by $O(n^{\alpha_d\log n})$ from above, for a constant $\alpha_d$ that monotonously decreases towards~$3/2$ as $d$ grows. 
Specifically, for $d=2$, they prove that the optimal number of comparisons lies within the range $(C_1 n^{3/2} , C_2n^{5/3}\log n)$ for some constants $C_1,C_2$.
Alon, Azar, and Vishkin~\cite{AAV86} improved the lower bound to~$\Omega(n^{1+1/d}(\log n)^{1/d})$. Alon and Azar~\cite{AA87,AA88} lower-bounded the \emph{average} number of comparisons by~$\Omega(dn^{1+1/d})$, for any $d$-round algorithm with $d\le \log n$.
They also improved the upper bound to $O(n^{1+1/d}\log n)$ for a fixed~$d$, and to $dn^{1+O(1)/d}$ for any $d\le \log n$.
Bollob{\'a}s and Rosenfeld considered a relaxed sorting task, where the  relative  order of $\eps n^2$ pairs might still be unknown at the end. They showed that by performing $C_{\eps} n^{3/2}$ comparisons, one can learn the order of $\binom{n}{2}-\eps n^2$ pairs, where $\eps\to0$ as $C_{\eps}\to\infty$.
In contrast to the above existential bounds for $2$-comparator based algorithms, our work provides  \emph{explicit} sorting algorithms. 
Our algorithms are efficient, they utilize $t$-comparators (allowing large values of~$t$) and are asymptotically optimal, with respect to the above bounds. 

Other related tasks were also considered in the literature. Alon and Azar~\cite{AA88} gave bounds on the  number of comparisons required for approximate sorting and for selecting the median. 
Braverman, Mao, and Weinberg~\cite{BMW16} considered the task of selecting the $k$-rank item, in a single  round (and multiple rounds), and of partitioning an unordered array into the $k$-top and $(n-k)$-bottom elements, in a single round. Their algorithms also work in the noisy-comparison setting, where each comparison is correct with probability~$2/3$.
Braverman, Mao, and Peres~\cite{BMP19} extended the above results and gave an algorithm sorting the $k$-top elements in small number of rounds ($d=1,2$ and $d\ge3$). They also give lower and upper bounds for this task, both in the noiseless and noisy-comparison setting.

A related approach for sorting is via \emph{sorting networks}~\cite{Batcher68,AKS83,leighton85} and in particular, sorting networks of $t$-comparators, a task that was raised by  Knuth~\cite[Question 54 in Section 5.3]{knuthVol3-Sorting} and examined in~\cite{AKS83,PP89,Chiang01,DKP23}. 
These are fixed networks of comparators with $n$ inputs (each element is an input) and $n$ outputs (the sorted elements). %
One main difference between our $d$-round sorting and a sorting network is that in the latter,
each element appears exactly once as an input. Then, any comparator that gets this element as an input must appear in a different ``round''. 
However, in a sorting algorithm, it is possible to give the same element to multiple comparators at the same round, and then form the total order out of the outcomes of all comparators (for instance, by the approach in Section~\ref{sec:FromCompsToSorting}).

Distributed sorting has appeared in the literature before, but it had a different meaning than the distributed sorting we consider here.
Wegner~\cite{WEGNER84} and 
Rotem, Santoro, and Sidney~\cite{RSS85} considered the task of moving records around in a distributed network, so that they end up in a sorted manner (i.e., records that end up at the first site have keys which are strictly smaller than the records in the second site). These works mainly focused on the number of exchanged messages.
We also briefly mention parallel VLSI sorting algorithms, e.g., \cite{TK77,SS86,SSS86,KK92,OZ96}. Here the common setting is of $n\times n$ parallel processors, usually connected as a two-dimensional grid. Each processors holds one element at any given time and can transfer the element to a neighboring processors. The goal is that the elements will end up in an ordered alignment, i.e., the minimal element at the first processors, etc. This setting is somewhat similar to our case of $n=t^2$, if we think of a row or a column of processors as a single unit that can re-order the elements in that row or column according to their rank.
Another sorting variant was considered by 
Patt-Shamir and Teplitsky~\cite{PT11} (building on~\cite{LW11}). 
Here, each computer starts with $n$ records and needs to output their rankings in the global order of all $n^2$ records. Also unlike our task, each computer can sort any number of records that it holds (i.e., it is not limited to being a $t$-comparator).

As mentioned above, randomized \emph{quicksort} with $t$-comparators was given by Beigel and Gill~\cite{beigel1990sorting}. This algorithm features an optimal number of comparators, albeit it employs a large number of rounds, $d=O(\log_t n)$. \emph{Mergesort} with $t$-comparators is given in~\cite{SZW14}, and \emph{cubesort} with $t$-comparators is presented in~\cite{CS92}.

\subsection{Organization}
We formally state the problem of sorting with $t$-comparators, setting the relevant notations in Section~\ref{sec:prelim}. We discuss one-round deterministic sorting in Section~\ref{sec:detrministic}. Our optimal 2-round randomized algorithm and its analysis using the  binary tree of deferred randomness, can be found in Section~\ref{sec:random}. In Appendix~\ref{app:simulations} we provide some simulations comparing our 2-round randomized algorithm with the state-of-the-art $O(\log_t n)$-round $t$-quicksort algorithm, showing that the latter has in fact an expected number of rounds strictly larger than~4 when~$n=t^2$.
In Section~\ref{sec:FromCompsToSorting} we describe how to apply our algorithms in distributed settings, where each comparator is an independent device in the network.
Conclusions and some future directions are presented in Section~\ref{sec:conclusion}.

\section{Preliminaries}
\label{sec:prelim}
\textbf{Notations.}
For a positive integer $n$, we let $\setO{n}$ denote the set $\{1,2,\ldots, n\}$.
All logarithms are taken to base 2 unless otherwise noted. We say that an event happens \emph{with high probability} in some parameter (usually, in the number of elements~$n$), if the event occurs with probability at least $1-1/n^c$ for some positive constant~$c\ge1$.

\bigskip\noindent
\textbf{Problem Statement.}
The elements are $A=\{a_1, a_2, \ldots, a_n\}$. Each element has a  value $\val(a_i)\in [n]$. We assume that all values are unique, so that for any $i\ne j$, $\val(i)\ne\val(j)$, and all values in~$[n]$ are covered. 

A $t$-comparator is a device that gets $t$ elements $\{a_{i_1}, \ldots, a_{i_t}\}$ as an input, and outputs the respective order of their values. 
That is, it outputs a list $j_1, \ldots, j_t$ of indices, such that these  are a permutation of ${i_1}, \ldots, {i_t}$ and it holds that $val(a_{j_1}) \le \val(a_{j_2}) \le \cdots \le \val(a_{j_t})$. Note that it is allowed to give as an input the same element multiple times (hence the inequality in the $\val()$ values).

A \emph{round of sorting} is any assignment of elements to (possibly multiple) comparators. The output of a single round of sorting is defined to be the  output of all the comparators in that round, i.e., the relative order between any $t$ elements compared by some comparator.
\begin{definition}
\label{def:sorting}
\emph{Sorting $n$ elements in $d$ rounds via $t$-comparators} is performing $d$ rounds of sorting, where the assignment of round $i \le d$ depends on  the outputs of rounds $1,\ldots, i-1$. 
The assignment of elements to comparators is such that, for any possible assignment of values to the elements, there exists a \emph{single} total ordering of the $n$ elements that is consistent with all the outputs of the $d$~rounds.
\end{definition}

We will usually care about the number of $t$-comparators required to sort $n$ elements.
Let us denote $\OPT(n,t,d)$ the minimal number of $t$-comparators required to sort $n$~elements in $d$-rounds.
In this paper we will focus on small values of~$d$.
In particular, 
in Section~\ref{sec:detrministic} we analyze the case of deterministic sorting in $d=1$ rounds. In Section~\ref{sec:random} we discuss randomized sorting with $d=2$ rounds.

\section{Sorting \texorpdfstring{$n$}{n} elements in a single deterministic round}
\label{sec:detrministic}
In this section we analyze sorting $n$ elements with $t$-comparators in a single round. That is, we seek ways to assign elements to comparators that yield enough information to obtain a total-ordering of the elements. Since we restrict ourselves to a single round, we cannot adaptively select elements to compare based on previous result. Instead, all the assignments must be predetermined.

We begin with a few straightforward observations and facts.
The following lemma is probably a well known  folklore: if we are allowed to compare only pairs of elements ($t=2$) and the comparisons are non-adaptive ($d=1$), then \emph{all} pairs of elements must be compared in order to obtain the total-ordering of the $n$ element. 
\begin{lemma}
\label{lem:t2allPairs}
For $t=2$,   
sorting $n$ elements with $2$-comparators in $d=1$ rounds requires learning the relative order of each of the $\binom{n}{2} = \Theta(n^2)$ pairs of elements. Thus, $\OPT(n,t=2,d=1)=\binom{n}{2}$.
\end{lemma}
\begin{proof}
Otherwise, there are two elements $a_i,a_j$ that are not compared against each other. Let the two minimal elements (in the ranking) be $a_i,a_j$, respectively. Switching their relative order (i.e., letting the minimal elements be $a_j,a_i$, respectively) will not change the outputs of any of the comparators. Hence, there are two total ordering consistent with all the outputs, contradicting the fact that this is a sorting of $n$ elements, Definition~\ref{def:sorting}.
\end{proof}
\begin{corollary}\label{cor:triv-lowerBound}
$\OPT(n,t,1) \ge \binom{n}{2}/\binom{t}{2}$.
\end{corollary}
\begin{proof}
The proof of Lemma~\ref{lem:t2allPairs} extends to larger comparators. If two elements are not being compared by some comparator, let them be of minimal value and exchange their relative order to end up with two consistent total ordering. Thus, $\OPT(n,t,1)$ must provide enough comparators to compare all pairs.

Each $t$-comparator gives the ranking of~$t$ elements among themselves. That is, it allows us to learn the (pair-wise) order between at most $\binom{t}{2}$ pairs of elements.
The statement immediately follows. 
\end{proof}
Note that 
\begin{align}
\label{eqn:LB}
    \frac{\binom{n}{2}}{\binom{t}{2}} = \frac{n(n-1)}{t(t-1)} \ge \frac{n^2}{t^2}-\frac{n}{t^2},
\end{align}
We can show that sorting with at most twice the amount of optimal comparators of Eq.~\eqref{eqn:LB} can be achieved for certain values of~$n,t$; sorting with at most three times the optimal is always possible. 
\begin{lemma}\label{lem:loose-upperbound}
When $(t/2) \mid n$,
$\OPT(n,t,1) < 2 {\binom{n}{2}}/{\binom{t}{2}}$.
Otherwise, $\OPT(n,t,1) < 3 {\binom{n}{2}}/{\binom{t}{2}}$.
\end{lemma}
\begin{proof}
    Assume $(t/2) \mid n$. Split the $n$ elements into $2n/t$ subsets of size $t/2$ each, $S_1,\ldots,S_{2n/t}$. Now, for any $i,j\in [2n/t]$ compare the elements in $S_i \cup S_j$ using a $t$-comparator. 
    It is immediate that any two elements will be compared in this process. The total number of comparators used is 
    \begin{align*}
        \binom{2n/t}{2} &= \frac12 \cdot \frac{2n}{t}\left(\frac{2n}{t}-1\right)
         =  \frac{2n^2}{t^2}-\frac{n}{t}.
    \end{align*}
    The above is clearly larger than twice Eq.~\eqref{eqn:LB}, by noting that 
    $n/t \ge 2n/t^2$ holds for $t\ge2$.

However, when $t/2$ does not divide $n$, we need one additional subset $S_{2n/t+1}$ for the leftovers. This results with a total of $2n^2/t^2+n/t$ comparators. When $t<\frac{-1+\sqrt{1+8n}}{2}$,  this is still within a factor~2 of ${\binom{n}{2}}/{\binom{t}{2}}$. Otherwise, it is easy to see that we are within a factor $3$  of the lower bound. Let us bound the ratio
\begin{align*}
\frac{\frac{2n^2}{t^2}+\frac{n}{t}}{\frac{n(n-1)}{t(t-1)}} = \frac{t-1}{t}\cdot\frac{2n+t}{n-1}.
\end{align*}
The right hand side monotonically increases in~$t$, and obtains its maximal value at~$t=n-1$. This yields
\begin{align*}
 \frac{n-2}{n-1}\cdot\frac{3n-1}{n-1}.
\end{align*}
This function monotonically increases in~$n$ (as can  easily be seen from its derivation) and has a limit of~$3$ as $n\to\infty$. 
\end{proof}

\subsection{The case of a large \texorpdfstring{$t$}{t}}
Let us now give optimal sorting assignments with $d=1$ for the case of a large comparator, $t=\Omega(n)$.
To demonstrate the basic idea, assume $t=n-1$. We argue that three comparators are sufficient in this case, which makes the bound in Lemma~\ref{lem:loose-upperbound} tight for~$n\ge 9$.
First, we compare
$\{a_1,\ldots, a_{n-1}\}$ which gives a total-ordering for all elements but the last element,~$a_n$, so we need to compare $a_n$ with all the other elements. This  can be done with by employing two additional comparators,
e.g., comparing $\{a_n, a_2\ldots, a_{n-1}\}$ and $\{a_n, a_1, \ldots, a_1\}$. Note that the second comparator is substantially under-utilized. This means that we could still perform sorting with only three comparators even for smaller values of~$t$.

\begin{lemma}
\label{lem:t-twothirds-n}
For any $t \ge \tfrac23 n$, sorting $n$ elements in a single round can be done with three comparators.
\end{lemma}
\begin{proof}
The inputs to the three comparators are 
(1) $\{a_1, \ldots, a_t\}$, (2)
$\{a_n, a_{n-1}, \ldots, a_{t+1} \textbf{\ ,\ } a_1, a_2, \ldots, a_{\lceil t/2\rceil }\}$ and (3)
$\{a_n, a_{n-1}, \ldots, a_{t+1}\textbf{\ ,\ } a_{\lceil t/2\rceil+1}, \ldots a_t\}$. Note that any two elements $a_i,a_i$ are being compared by some comparator, yielding all the information we need to obtain a single consistent total order of the elements.

Since $t\ge \tfrac23n$, the second and third comparators get each 
$(n - (t+1)+1) + \lceil \tfrac{t}2 \rceil 
\le \lfloor \tfrac32t \rfloor-t+\lceil \tfrac{t}2\rceil = t
$ elements as input. Note that the ceiling/flooring matters only when $t$ is odd. In this case $\tfrac32t$ is fractional and since $n$ must be an integer, we have $n \le \lfloor \tfrac32 t\rfloor$.
\end{proof}

The above three comparators construction is tight, as it is impossible to sort $n$ elements with only two comparators. The proof resembles the approach taken by Lemma~\ref{lem:t2allPairs} for the case of $t=2$.
\begin{lemma}
\label{lem:no2comps-possible}
For any $t<n$, 
sorting in one round cannot be achieved with two comparators.
\end{lemma}
\begin{proof}
By a pigeonhole principle, there must exist (at least) two elements $a_i,a_j$ that are not compared against each other.  We make it so 
$\forall k\in [n] \setminus \{ i,j\}, \val(a_i)<\val(a_k) \text{ and} \val(a_j)<\val(a_k)$.
Then, it is impossible to determine which one of $a_i,a_j$ is the minimal element. 
Specifically, setting $\val(a_i)<\val(a_j)$ gives the same comparator outputs as the case where $\val(a_j)<\val(a_i)$. This follows since they both are lower than any other element and no comparator has both of them as input. Then, there exists two total ordering consistent with the output of the comparators: one with $\val(a_i)<\val(a_j)$ and the other with $\val(a_j)<\val(a_i)$, contradicting Definition~\ref{def:sorting}.
\end{proof}

\subsection{Minimal sorting for a variety of parameters via design theory}
Recall the proof of Corollary~\ref{cor:triv-lowerBound}. It implies that every two elements must be compared against each other. 
This leads us to defining \emph{minimal} sorting as follows.

\begin{definition}
    Sorting is said to be \emph{minimal} if equality holds in the equation in Corollary~\ref{cor:triv-lowerBound}.
\end{definition}
That is, minimality is obtained when every two elements are compared against each other \emph{exactly} once, and all the $t$-comparators are fully utilized. Then on the one hand there is no redundancy, and on the other hand all computational resources are fully used.
Note that \emph{optimality} means the minimal number of comparators needed to get all pairs compared against each other exactly once, but without requiring that all comparators are fully utilized. 

While minimality implies optimality, the other direction does not hold.
As demonstrated above for $2n/3 \le t<n$, optimality is obtained with 3 comparators. However, minimality is not obtainable in this case. For instance, when $n=10$, and $t=7,8,9$ we have $\binom{n}{2}/\binom{t}{2}\in[1\tfrac14,2\tfrac17]$, but, as we proved, exactly 3 comparators are necessary in all these cases, i.e., some comparator must be under-utilized regardless of the sorting algorithm.

\begin{tcolorbox}[boxrule=1pt,leftrule=3mm]
Our findings on minimal sorting presented here were not developed in a linear fashion, as is common in research. 
We first came up with a minimal sorting design for $t$ being a prime and $n=t^2$. We then extended this result to $t$ being a power of a prime. 
Both these constructions appear in Appendix~\ref{APP:minConstrustions}. 
It was only later that we tried to prove impossibility bounds and realized that minimal sorting is equivalent to a combinatorial structure known as a 2-Steiner system. 
In the following, we jump directly to the end of this line of thought, illustrate the equivalence, and present minimal constructions for some special cases of the set of parameters.
\end{tcolorbox}

\begin{definition}[A Steiner System]
\label{def:steiner}
    A \emph{Steiner System} with parameters $0<c< t<n$, denoted $S(c,t,n)$, is a set~$\P$ of $n$ elements (we will call \emph{points}) and a set~$\L$ of objects (we will call \emph{lines}), where each line is a subset of $t$ points and it holds that any subset of~$c$ points is contained in exactly a single line.
\end{definition}

Corollary~\ref{cor:triv-lowerBound} and the discussion above imply the following.
\begin{theorem}
\label{thm:equivalence:sortIsSteiner}
    The Steiner system $S(2,t,n)$ is equivalent to a minimal sorting of $n$ elements via $t$-comparators.
\end{theorem}
\begin{proof}
    Immediate from definitions. Every point is an element to sort, every line is a single comparator. Since any two points are contained exactly in a single line and since every line contains exactly $t$~points, we obtain minimality.
\end{proof}
The above equivalence allow us to use known results about $S(2,t,n)$ to deduce cases for which minimal sorting is possible. The following is an immediate corollary of the known state-of-the-art about Steiner systems with $c=2$, see e.g., \cite{Hughes_Piper_1985,GG94,reid2012steiner}.

\begin{theorem}\label{thm:minimal-sorting-params}
~
\begin{enumerate}
    \item Let $t$ be a power of a prime. Minimal sorting of $n=t^2$ elements  is possible by employing $t^2+t$ many $t$-comparators.
    \item Let $t-1$ be a power of a prime.
    Minimal sorting of $n=t^2-t+1$ elements 
    is possible by employing $t^2-t+1$ many $t$-comparators.
\end{enumerate}
\end{theorem}
\begin{proof}
    (1) Follows from the fact that every field of size $t$ implies a Steiner system $S(2,t,t^2)$ (an Affine plane), see \cite[Section~3.2]{Pascoe18}. (2) Follows from the fact that every field of size $t-1$ implies a Steiner system $S(2,t,t^2-t+1)$ (a Projective plane), see \cite[Section~4.5]{Pascoe18}. We note that both constructions are explicit.
\end{proof}

The equivalence stated in Theorem~\ref{thm:equivalence:sortIsSteiner} also yields some impossibilities on minimal sorting. It is well known that the Steiner system $S(2,6,36)$ does not exist. This problem, stated originally as a question about Latin Squares and known as the 36 officers problem, dates back to Euler~\cite{euler1782} and was proven impossible by Terry~\cite{Tarry1900}.
Bruck and Ryster~\cite{BR49} extended this result and proved that Steiner systems of  many other orders are also impossible. 
\begin{corollary}[\cite{Tarry1900,BR49}]
    Minimal sorting of $n=t^2$ elements (i.e., with exactly $t+t^2$ many $t$-comparators) is impossible for infinitely many values of~$t$.
\end{corollary}

Despite decades of research, a full characterization of values of~$t$ that admit a $S(2,t,t^2)$ system does not exist. 
In 1975, Willson~\cite{WILSON75III} showed that for any $t$, a Steiner $S(2,t,n)$ system exists if and only if $t \mid n$ and $t(t-1) \mid n(n-1)$, \emph{except for finitely many values of~$n$}. This implies the following corollary
\begin{corollary}
For any $t$ and large enough integer~$c$, minimal sorting of $n=t^c$ elements is possible with $\OPT(n,t,1)=\binom{n}{2}/\binom{t}{2}$ many $t$-comparators.
\end{corollary}
Indeed, for any $c\ge 1$ we have that $t \mid t^c$ and $(t-1) \mid (t^c-1)$ since $t^c-1 = (t-1)(t^{c-1}+t^{c-2}+\cdots+1)$. 
Our composition theorem, which is given in the next section (Lemma~\ref{lem:composition}), gives explicit construction for some values of~$n,t$. Finding explicit constructions for other values remains open.

\subsection{A Composition Theorem}
The above Theorem~\ref{thm:minimal-sorting-params} applies only to the cases where $n=t^2$ or $n=t^2-t+1$ (for certain values of~$t$). An interesting question is how to obtain a single-round sorting for other values of $t$ and $n$, e.g., for $n=t^c$ elements, with~$c\ge3$. 
We partially answer this task by constructing a $t^2$-comparator out of an optimal number of $t$-comparators. Operating recursively on larger $n$'s, this approach leads to the following theorem.
\begin{lemma}
\label{lem:composition}
    Let $t$ be power of a prime and let $n=t^{2^k}$ for some $k\in \N$. 
    Then, minimal sorting of $n$ elements with $t$-comparators is possible and employs   $\OPT(n=t^{2^k},t,1)=\binom{n}{2}/\binom{t}{2}$ many $t$-comparators.
\end{lemma}
\begin{proof}
    We prove that minimal sorting is possible by induction on~$k$. The base case, $k=1$ is given by Theorem~\ref{thm:minimal-sorting-params}(1).

    For the induction step, assume we can sort $n'=t^{2^{k-1}}$ elements using $\binom{n'}{2}/\binom{t}{2}$ many $t$-comparators. We show how to sort $n=t^{2^k}$ elements with exactly $\binom{n}{2}/\binom{t}{2}$ $t$-comparators. Since $n'$ is a power of a prime, Theorem~\ref{thm:minimal-sorting-params} provides us a optimal (minimal) way to sort $n$ elements using $n'$-comparators. Each $n'$-comparator can be implemented via an optimal (minimal) number of $t$-comparators, by induction. The total number of $t$-comparator thus required to sort $n$ elements is
    \begin{align*}
        \frac{\binom{n}{2}}{\binom{n'}{2}} \cdot 
        \frac{\binom{n'}{2}}{\binom{t}{2}}= 
        \frac{\binom{n}{2}}{\binom{t}{2}}\text{,}
    \end{align*}
    and this quantity is minimal due to Corollary~\ref{cor:triv-lowerBound}.
\end{proof}

As a corollary, the above composition theorem implies an explicit construction of a $S(2,t,t^{2^k})$ system for $t$ a power of a prime and all integers~$k>0$. 

\section{Optimally sorting \texorpdfstring{$n$}{n} elements in \texorpdfstring{$d=2$}{d=2}  randomized rounds}
\label{sec:random}

In Section~\ref{sec:detrministic}, we studied optimal deterministic sorting in $d=1$ rounds. We now wish to turn to the case of $d=2$ rounds, trading-off one additional round for fewer comparisons.
We study the \emph{randomized} case since it allows us to reduce the number of  comparisons considerably. Since for $d=1$ we have already obtained an optimal deterministic solution, it makes sense to discuss randomized algorithms for $d>1$. 
As randomized sorting with $O(\log n)$ comparators are well-known \cite{quicksort,beigel1990sorting}, we wish to keep the number of rounds small, and focus on the case of~$d=2$.
We design a fast randomized $t$-sorting algorithm, which is asymptotically optimal in the number of $t$-comparators used, restricted to algorithms with $d=2$ rounds. In certain cases, for instance when $n=t^2$, the asymptotic number of $t$-comparators is optimal even without the round restriction. 
We discuss lower bounds on the number of $t$-comparators required for sorting in Section~\ref{sec:rand:LB}.
In Section~\ref{sec:rand:n-tsquare} we consider the special case of $d=2$ and $n=t^2$ and in Section~\ref{sec:rand:general} we consider the more general case of $d=2$ and arbitrary $n$ and~$t$. Our main result is Theorem~\ref{thm:main-rand-informal}, which we now recall.

\RESTATABLEmainRandomized*

\subsection{Lower bounds}
\label{sec:rand:LB}
Before describing our algorithms, let us recall the lower bound on the number of $t$-comparators, by Beigel and Gill~\cite{beigel1990sorting}.
\begin{theorem}[\cite{beigel1990sorting}]
\label{thm:OPT-rand}
Sorting $n$ elements requires utilizing at least
\(
\frac{\log (n!)}{\log (t!)}=\frac{n\log n}{t\log t}(1+o(1))
\)
many $t$-comparators.
\end{theorem}
The proof stems from the fact that $\log (n!)$ bits of information are required to sort $n$ elements, and that each comparator gives $\log (t!)$ bits of information. See Section~II in~\cite{beigel1990sorting}.

The above lower bound allows any number of rounds. Alon and Azar~\cite{AA87}  analyzed the average number of $2$-comparators required to sort $n$ elements in $d$ rounds and proved the following.
\begin{theorem}[\cite{AA87}]
\label{thm:avg-comps-Kround-LB}
Sorting $n$ elements in $d\le \log n$ rounds, requires utilizing 
at least
 \(   \Omega(d n^{1+1/d})
 \)
many $2$-comparators \emph{on average}.
\end{theorem}
The above theorem could be used to derive lower bounds on sorting with $t$-comparators. Recall that each $t$-comparator compares at most $\binom{t}{2}$ pairs of elements. Then, the  following lower bounds on the avergae number of $t$-comparators required in any randomized sorting is immediate.
\begin{corollary}
    \label{cor:LBn32/t-rand}
    Sorting $n$ elements in $d\le \log n$ rounds, requires utilizing 
at least
 \(
 \Omega(d n^{1+1/d} / t^2)
 \)
many $t$-comparators \emph{on average}.
\end{corollary}
Because any average-case lower bound is also a worst-case lower-bound,
if we plug in $d=2$ in the above corollary, we obtain that our algorithm with $O(n^{3/2}/t^2)$ many $t$-comparator when $t<\sqrt{n}$, is asymptotically tight.

\subsection{The simple special case of \texorpdfstring{$n=t^2$}{n=t\texttwosuperior}}
\label{sec:rand:n-tsquare}
In this section we present Algorithm~\ref{alg:rand-2round}, which performs $t$-sorting of $n=t^2$ elements in two rounds and utilizes $O(t)$ many $t$-comparators. Note that by Theorem~\ref{thm:OPT-rand}, this is asymptotically tight, even without the restriction to $d=2$ rounds.
Although our Algorithm~\ref{alg:rand-general} and Algorithm~\ref{alg:rand-general2}
described in Section~\ref{sec:rand:general} are strictly more general, as they apply to any $n,t$, for pedagogical reasons we first introduce the simplified and very natural Algorithm~\ref{alg:rand-2round} that assumes the special case of~$n=t^2$.

\begin{algorithm}[htp]
\caption{A randomized 2-round sorting of $n=t^2$ elements with $O(t)$ many $t$-comparators}
\label{alg:rand-2round}
\begin{algorithmic}[1]
\Statex
\Statex \textit{\textbf{Round 1:}}
\State Let $P$ be a set of $t/2$ elements from $A$, each sampled uniformly and independently from~$A$. 
\label{line:sample}
\State Partition $A$ into subsets $A_1,\ldots, A_{k}$ of size $t/2$ each.
\ForAll {$i\in[k]$}
    \State Input $P\cup A_i$ into a comparator.
        \Comment{$k$ comparators}
\EndFor

\Statex
\Statex \textit{\textbf{Round 2:}}
\State 
Let $P=(p_1, \ldots, p_{t/2})$ be the ordered elements in $P$. For $1\le i \le t/2-1$, set $S_i$ to contain all the elements which are greater than $p_i$ but lower than $p_{i+1}$. Set $S_0$ to be all the elements lower than $p_1$ and $S_{t/2}$ be all the elements greater than $p_{t/2}$.
\ForAll {$0 \le i \le t/2$}
    \State Sort $S_i$ via Lemma~\ref{lem:loose-upperbound}.
        \Comment{at most $\sum_i 3|S_i|^2/t^2$ comparators}
\EndFor
\end{algorithmic}
\end{algorithm}

Recall our notations, where we wish to sort a set of $n=t^2$ elements, denoted  $A=\{a_1, \ldots, a_n\}$.
We assume that $t$ is even and that $(t/2) \mid n$, and set $k = n / (t/2)$.
The algorithm works as follows. In the first round, we first sample $t/2$ elements uniformly from~$A$. These will be ours ``pivots''. We then take the remaining elements of $A$ and compare them to the pivots. That is, we split the remaining elements into  $n/(t/2)-1$ disjoint subsets of size~$t/2$. We  input each subset to a $t$-comparator together with (all) the $t/2$~pivots. After this step, for each element in~$A$ we know its relative position with respect to the pivots. 
Since we used the same pivots in each comparator, we can see the first round as the pivots splitting $A$ into $t/2+1$ disjoint ``buckets'' such that all the elements in one bucket are strictly smaller (or strictly larger) than all elements in any other bucket.
In the second round of the algorithm, we sort each bucket separately.

The first step utilizes $n/(t/2)=2t$ comparators, one for each subset of~$A$. In the second part, the number of comparators in use depends on the size of the buckets we need to sort, which is a random variable determined by the pivots we sample in the first round.
In expectation, each bucket is of size approximately\footnote{\label{fn:boundingExpectation}To bound the expected size of each bucket, consider the sorted array of elements 
and uniformly select $t$~pivots. 
Connect the beginning of the array to its end to form a cycle. Now consider all intervals between the pivots. The expected sum of the intervals, is roughly~$n$. By linearity of expectation, we can  consider disjoint `chunks' of intervals, each composed of $t$ consecutive intervals.
By symmetry, the expected lengths of all chunks are the same. Thus, each chunk must be, in expectation, about $n/t$ elements long (ignoring constants).
} %
$n/(t/2+1)\approx 2\frac{n}{t}$. 
If we assumed that the number of elements per bin is tightly concentrated around its means, then we could deduce that sorting a single bucket   using Lemma~\ref{lem:loose-upperbound} would take $O(\frac{n^2}{t^4})=O(1)$ comparators, and summing up over all $t/2+1$ buckets results in  $O(\frac{n^2}{t^3})=O(t)$ comparators overall, in expectation. 

However, we cannot make such an assumption, since, while each bucket has $\approx 2\frac{n}{t}$ elements in expectation, there might be very large buckets, with, say, $O(\frac{n}{t}\log n)$ elements. 
Our analysis (which we perform only to the general case, in Section~\ref{sec:rand:general} below),  is somewhat more intricate and shows that the event of a large bucket is rare enough so that amortizing across all the buckets, our algorithm still takes~$O(t)$ comparators with high probability.

\subsection{The general case: supporting any \texorpdfstring{$n,t$}{n,t}}
\label{sec:rand:general}

Algorithm~\ref{alg:rand-2round} can be executed with any $n,t$. 
The problem is that this would come at a very high cost (measured in the number of $t$-comparators used).
The main reason for this high cost is that Algorithm~\ref{alg:rand-2round} has a tradeoff between the costs of the different rounds:
the cost of the first rounds is~$O(\frac{n}{t})$ and the 
cost of the second is~$O(\frac{n^2}{t^3})$. 
While these two costs equal~$O(t)$ for~$n=t^2$, for arbitrary $n$ and~$t$ these costs are no longer balanced and one of the rounds would have a relatively high cost. 
The idea behind Algorithm~\ref{alg:rand-general} depicted below,\footnote{Algorithm~\ref{alg:rand-general} is identical to Algorithm~\ref{alg:rand-general-intro} described in the introduction and repeated here for convenience.} is to balance the costs of the phases, by carefully choosing the size of the pivot set and, as a result, the expected sizes of the buckets they yield.

\begin{algorithm}[htp]
\caption{A randomized 2-round sorting for any $n,t$ with $t\leq \sqrt{n}$}
\label{alg:rand-general}
\begin{algorithmic}[1]
\Statex
\Statex \textit{\textbf{Round 1:}}
\State Let $P$ be a set of $m=\sqrt{n}$
elements from $A$, each sampled uniformly and independently from~$A$.\label{line:sample:B}%
\State Partition $A \setminus P$ into subsets $A_1,\ldots, A_{k}$ of size at most $m$ each. \Comment{$k=\lceil (n-m)/m\rceil$}
\ForAll {$i\in[k]$}
    \State Sort $P\cup A_i$ via  Lemma~\ref{lem:loose-upperbound}.
\EndFor

\Statex
\Statex \textit{\textbf{Round 2:}}
\State 
Let $P=(p_1, \ldots, p_{m})$ be the ordered elements in~$P$. For $1\le i \le m-1$, set $S_i$ to contain all the elements which are greater than $p_i$ but lower than $p_{i+1}$. Set $S_0$ to be all the elements lower than $p_1$ and $S_{m}$ be all the elements greater than $p_{m}$.
\ForAll {$0 \le i \le m $}
    \State Sort $S_i$ via Lemma~\ref{lem:loose-upperbound}.
\EndFor
\end{algorithmic}
\end{algorithm}

Assume that the first round randomly selects $m$ pivots, which we denote by the set~$P$.
In order to ``bucket'' the $n$ elements according to the pivots we need to compare them all with all the pivots. To that end, we split the set $A$ into subsets $A_1,\ldots, A_k$ of size~$m$ (maybe except for the last subset), and compare each subset with the pivots. In contrary to Algorithm~\ref{alg:rand-2round}, we can no longer input $A_i \cup P$ into a $t$~comparator. Instead, we need to implement a $2m$-comparator out of $t$-comparators. We do so via Lemma~\ref{lem:loose-upperbound}, at the cost of $8 m^2/t^2$ many $t$-comparators for a single simulated $2m$-comparator.

Let us now analyze the expected cost of Algorithm~\ref{alg:rand-general}.
To calculate the cost of the first round, note that we now need $n/(2m)$ many (simulated) $2m$-comparators each costing us $O(m^2/t^2)$ 
 many $t$-comparators.
Thus, the first round results in a total cost of $O(\frac{nm}{t^2})$.
The expected cost of the second round is given as follows: 
since the set of pivots is sampled uniformly,  the expected size of each bucket is 
$\approx 2n/t$
(Footnote~\ref{fn:boundingExpectation}).
Oversimplifying again
and assuming the number of elements per bin is tightly concentrated (which is not necessarily true for each bin), we get the following. 
By Lemma~\ref{lem:loose-upperbound}, each one of the $m+1$ buckets costs~$O((n/m)^2/t^2)$ comparators in expectation.
Overall, the expected cost in the second round is~$O(\frac{n^2}{m t^2})$.
Summing the costs of the two rounds, the expected cost of Algorithm~\ref{alg:rand-general}
is~$O(\frac{nm}{t^2}+\frac{n^2}{m t^2})$.
Interestingly, this value is minimized when $m=\sqrt{n}$, irrespective of~$t$. 
In the reminder, we simply set $m=\sqrt{n}$, and the cost becomes $O(\frac{n^{3/2}}{t^2})$.

\subparagraph{The case of
\texorpdfstring{$t > \sqrt{n}$}{t > sqrt(n)}.}

The above analysis needs a little tweak to support the case of $t> \sqrt{n}$. In this case, the number of comparators-per-bucket given by the terms $O(m^2/t^2)$ and $O(n/m^2t^2)$ for the first and second round, respectively, is lower bounded by a single comparator, and thus should read $\max\left\{1, O(m^2/t^2)\right\}$ and $\max\left\{1, O(n/m^2t^2)\right\}$, respectively. Therefore, the choice of parameters needs to be adjusted.
In the following we show a selection of parameters that optimize the case of $t>\sqrt{n}$, which yields Algorithm~\ref{alg:rand-general2}. We only give here a sketch of the (simplified) expected cost analysis, since the precise high-probability analysis will be based on the same ideas presented in Section~\ref{sec:random:analysis} below for the case where $t\le \sqrt{n}$.

In Algorithm~\ref{alg:rand-general2}, We set the number of pivots to be $\tilde m=\lceil n/t\rceil$, and group the rest of the elements into subsets $\{A_i\}$ of size $t$ each (instead of size $\tilde m$). We then continue with the sorting as before. 

In the first round of the algorithm, we sort $k = O(n/t)$ sets, each of size $t+{\tilde m} =O(t)$. 
Thus, by Lemma~\ref{lem:loose-upperbound} sorting each such bucket can be done using $c'=O(1)$ comparators resulting in $c'k=O(n/t)$ comparators in total.
In the second round,  each $S_i$ has $O(n/\tilde m)=O(t)$ elements, in expectation. 
Assuming again our oversimplification that the number of elements in each bin is tightly concentrated around its mean, we get by Lemma~\ref{lem:loose-upperbound} that sorting each $S_i$ takes $O(1)$ comparators. 
Since there are ${\tilde m}+1$ such sets, the total number of comparators used in the second round is also bounded by~$O(n/t)$. 

\begin{algorithm}[htp]
\caption{A randomized 2-round sorting for any $n,t$ with $t > \sqrt{n}$}
\label{alg:rand-general2}
\begin{algorithmic}[1]
\Statex
\Statex \textit{\textbf{Round 1:}}
\State $\tilde m=\lceil n/t\rceil$
\State Let $P$ be a set of $b$
elements from $A$, each sampled uniformly and independently from~$A$.%
\State Partition $A \setminus P$ into subsets $A_1,\ldots, A_{k}$ of size at most $t$ each. \Comment{$k=\lceil (n-\tilde m)/t\rceil$}
\ForAll {$i\in[k]$}
    \State Sort $P\cup A_i$ via  Lemma~\ref{lem:loose-upperbound}.
\EndFor

\Statex
\Statex \textit{\textbf{Round 2:}}
\State 
Let $P=(p_1, \ldots, p_{\tilde m})$ be the ordered elements in~$P$. For $1\le i \le {\tilde m}-1$, set $S_i$ to contain all the elements which are greater than $p_i$ but lower than $p_{i+1}$. Set $S_0$ to be all the elements lower than $p_1$ and $S_{\tilde m}$ be all the elements greater than $p_{\tilde m}$.
\ForAll {$0 \le i \le {\tilde m} $}
    \State Sort $S_i$ via Lemma~\ref{lem:loose-upperbound}.
\EndFor
\end{algorithmic}
\end{algorithm}

\medskip
In Section~\ref{sec:random:analysis} we formally analyze the number of comparators used by these schemes (without the oversimplifying assumption) and show that it is concentrated around the stated value, i.e., we prove Theorem~\ref{thm:main-rand-informal}.
As mentioned above, we  only analyze  Algorithm~\ref{alg:rand-general} since the analysis of Algorithm~\ref{alg:rand-general2} is analogous.
We stress again that the expected analysis presented above is oversimplified.
Further, even with a simple and straightforward expected analysis, the dependencies of the events make it difficult to obtain high-probability concentration bounds, i.e., bounds that hold except with a polynomially small  probability.  
Section~\ref{sec:random:analysis} is devoted to the development of tools that will allow us to obtain high-probability concentration bounds.

\subsection{Analysis: Obtaining high probability bounds}
\label{sec:random:analysis}
We urge the reader to first read the high-level description of our analysis, presented in Section~\ref{sec:highlevel}. We make use of the setting introduced there and  briefly recall the  main idea and notations. 
We define a balls-into-bins process, where we throws $m$ balls, uniformly at random, into $m/c$ bins; recall we substitute $m=\sqrt{n}$ throughout. We throw all the balls together, but draw, at each step, only one bit of randomness towards the string that describes their designated bin. This can be seen as a
\emph{binary tree of deferred randomness}, where the balls begin at the root node, and at every step each ball decides, with equal probability, to which children to continue. The tree has depth $T=\log(\sqrt{n}/c)$ where the root is at level~0.

Our goal is to bound the lengths of consecutive stretches of empty bins. 
For each node~$u$, we define $\COST(u)$ as the number of comparators needed in order to sort the elements associated with the bins below the node~$u$. 
In Particular, since each bin represents $c\sqrt{n}$ elements, 
sorting the elements in $s$~consecutive bins requires at most $ 3\binom{s c\sqrt{n}}{2}/\binom{t}{2}\leq 4 s^2 c^2 n/t^2$ comparators (Lemma~\ref{lem:loose-upperbound}).
For convenience, we will normalize the cost by a factor of $4 c^2 n/t^2$, so that the function $\COST(u)$ represents a \emph{normalized} cost. That is, we will have $\COST=s^2$ to indicate a number of $s^2 \cdot 4 c^2 n/t^2$ comparators.

One additional notation we need is the \emph{super-bin}: instead of dealing with each bin separately, we 
group together sets of $\frac{8}{c}\log n$ consecutive bins into one super-bin. The reason behind this definition is that in the upper layers of the tree, with high probability, the balls are distributed nicely, and our focus will be on the lower levels of the tree, namely, when nodes  have approximately $\frac{8}{c}\log n$  bins  (or, a super-bin) under them. 
Let us bound the number of balls~$L$ that are expected to reach a specific super-bin, with high probability.
Up to rounding, there are $(\sqrt{n}/c)/(\frac{8}{c}\log n)=\frac{\sqrt{n}}{8\log n}$ super bins and hence, each specific super-bin 
gets $\E{L}= \sqrt{n} /\frac{\sqrt{n}}{8\log n}=8\log n$ balls, in expectation.
Then, thinking about $L$ as the sum of $m=\sqrt{n}$ random variables indicating whether the $i$-th ball reached our super-bin or not, 
we can use a Chernoff bound with $\delta=7/8$, to have a concentration bound on the number of balls that end up in that super bin. 
\begin{align*}
\Pr{L \notin  [\log n, 15\log n] } = \Pr{L \notin \E{L}(1 \pm \delta) } \leq 
2\exp\left(- \frac{\delta^2 \E{L}}{3}\right)
< 2\exp\left(- 2\log n\right)
< 
2/n^2
\text{.}
\end{align*}

Let us now focus on the level of the tree where each node has $\frac8{c}\log n$ associated bins below it. We name this level by~$N_0$, to ease notation. 
It is easy to verify that $N_0=T-\log(8\log(n)/c)$. 
The argument above shows that with high probability, each node at level $N_0$ has a number of balls in the range $[\log n,15\log n]$.
Further, via a union bound on the 
$\frac{\sqrt{n}}{8\log n}<n$ 
nodes in level~$N_0$, the event that \emph{all} these nodes have a number of balls in the range  $[\log n,15\log n]$ holds with high probability, and we can condition on it in the following.

Let us now define in more detail the process during the last levels of the tree, which we rename as $N_0$ to $N_{\log(8\log(n)/c)}$. To ease notations we sometimes use $b=\log(8\log(n)/c)$.
Consider some node $u\in N_0$, and focus on the balls in this subtree alone.  
We start with all the balls at~$u$ and distribute them to $u$'s child with equal probability. We stop when we reach a leaf $z$ at level $N_{\log(8\log(n)/c)}$.
For any node~$u$, 
we denote by~$\ell_u$ the number of balls reaching~$u$.
For a non leaf~$u$, we define
$\mathcal{E}_u$ 
to be the bad event in which 
one of $u$'s children obtains strictly less than 
$h(\ell_u)=\frac{\ell_u}{2}-g(\ell_u)$ balls,
where $g(x)=\sqrt{2 x\log(x) }$.

If $\mathcal{E}_u$ occurs,  we stop the process here, for the subtree of~$u$. The analysis will show that we can charge this event with (amortized) cost of~$8\ell_u^2$, as this cost greatly surpasses the real cost for sorting the elements in the bins associated with~$u$.
Otherwise, if the event $\mathcal{E}_u$ did not occur in $u$ we continue distributing the balls. The (amotized) cost of a leaf can be shown to be at most~$8$. 

\medskip
Our goal in the rest of this section, is to bound, with high probability, the cost of the root of the binary tree of deferred randomness,~$\COST(\Root)$. 
We do so in three steps.
Since we already condition on the high probability event that no $\mathcal{E}_u$ happened in levels above~$N_0$, in the first step, we consider the (unlikely but illustrative) case, where $\mathcal{E}_u$ never happens in the subtree rooted by some $u\in N_0$. In this case, each leaf gets at least one ball (Lemma~\ref{lem:allgood}), and this would eventually imply the amortized cost per bin (and hence, for a the leaf~$u$) is a constant not bigger than~8.
The second step
deals with the case where
$\mathcal{E}_u$ happens at some node~$u$. 
This means that one child of~$u$ receives too few balls, possibly causing many of the bins below to be empty, i.e., without any ball reaching them. 
Unfortunately, 
bounding the cost of the bins associated with a single~$u$ is problematic, since this cost depends on the cost of two other nodes, whose bins are immediately to the left and to the right of the bins of~$u$.
In Lemma~\ref{lem:amortize} we bound the total cost of the tree, by summing over all the nodes~$v$ in which $\mathcal{E}_v$ happened and over all the remaining leaves above which  $\mathcal{E}_v$ did not happen. By doing so the dependencies between neighbors balance out and allow us to derive a bound on the total cost, $\COST(\Root)$, as a function of the bad events that happen throughout the process. From this total cost we can now derive an  \emph{amortized} cost per bad-event~$\mathcal{E}_u$, according to the depth of~$u$ (upper bounded by the number of balls~$\ell_u$).  
The last step uses Hoeffding's concentration bound (Theorem~\ref{thm:hoe}) to bound the number of bad events that happen at every level below~$N_0$. Incorporating this bound on the number of bad events with the amortized cost per bad-event, completes the proof of Theorem~\ref{thm:main-rand-informal} at the end of this section.

\begin{lemma}\label{lem:allgood}
Consider the binary tree of deferred randomness process  described above, and focus on the subtree rooted at some node $u$ at level~$N_0$.
Consider a path $u=u_0,u_1,u_2,\ldots, u_b$ from $u$ to one of its descendant leaves, and let $\ell_i$ be the number of balls reaching $u_i$.
If $\ell=\ell_0 \geq \log n$ and the bad events $\mathcal{E}_{u_i}$ do not occur for any $i \leq b$, 
then 
$\Pr{\ell_b \ge 1} =1$,
assuming a sufficiently large constant~$c$. 
\end{lemma}
\begin{proof}
Define $h^{(i+1)}(x)=h^{(i)}\left(h(x)\right)$ and $h^{(1)}(x)=h(x)$.
We have \begin{align} \label{eq:kuh}
\ell_b &\geq  h^{(b)}(\ell)= h^{(b-1)}\left(\frac{\ell}{2}-g(\ell)\right)
= h^{(b-2)}\left(\frac{\ell}{4}-\frac{g(\ell)}{2} -g(\ell/2-g(\ell))\right)\notag
\\&\geq
h^{(b-2)}\left(\frac{\ell}{4}-\frac{g(\ell)}{2} -g(\ell/2)\right)
\geq  \cdots \ge \frac{\ell}{2^b} - \sum_{i=0}^{b-1} \frac{g(\ell/2^i)}{2^{b-1-i}}.
\end{align}

Set $\alpha = \ell / 2^{b-1}$. We bound, 
\begin{align}\label{eq:esel} \sum_{i=0}^{b-1} \frac{g(\ell/2^i)}{2^{b-1-i}} &\stackrel{(1)}{=}
\sum_{j=0}^{b-1} \frac{g(\ell/2^{b-1-j})}{2^{j}} = 
\sum_{j=0}^{b-1} \frac{g(\alpha \cdot 2^j)}
{2^{j}}
= 
\sum_{j=0}^{b-1} \frac{\sqrt{2\alpha \cdot 2^j \log(\alpha \cdot 2^j)}}{2^j}
\notag
\\
&\leq 
\sqrt{2\alpha} \sum_{j=0}^{b-1} \frac{\sqrt{ j+\log(\alpha)}}{\sqrt{2^j}} \notag
 \stackrel{(2)}{\leq}
\sqrt{2\alpha} \sum_{j=0}^{b-1} \frac{\sqrt{j} + \sqrt{\log \alpha}}
{\sqrt{2^{j}}} 
\notag
\\
& \stackrel{(3)}{\leq}
\sqrt{2\alpha} (4.2 + (2+\sqrt{2})\sqrt{\log \alpha})
\stackrel{(4)}{\leq}
10 \sqrt{\alpha \log \alpha},
\end{align}
where transition~(1)  follows by substituting~$i$ with $j=b-1-i$, (2) follows since for any positive $x,y$ we have $\sqrt{x+y}^2\le x+2\sqrt{xy}+y = (\sqrt{x}+\sqrt{y})^2$,
transition~(3) follows since $\sum_{i=0} \sqrt{i/2^i} \leq 4.2 $ and $\sum_{i=0} \sqrt{1/2^i} \leq 2+\sqrt{2}$
and transition~(4) holds for  $\alpha\geq 5$.
We now continue bounding Eq.~\eqref{eq:kuh} using Eq.~\eqref{eq:esel}.
\begin{align} \label{eq:pferd}
\ell_b&\geq \frac{\ell}{2^b} - \sum_{i=0} \frac{g(\ell/2^i)}{2^{b-i}} \geq 2\alpha-10 \sqrt{\alpha\log \alpha} \geq 1,
\end{align}
where the last inequality holds for $\alpha\geq 200$.
Recall that $b =\log(8\log(n)/c) $ and $\ell \geq \log n$, hence, 
$\alpha=\ell/2^{b-1} \geq \log n/(8\log(n)/2c) = c/4$.
Thus if we set $c=800$, the claim holds.
\end{proof}

The above Lemma~\ref{lem:allgood} tells us that, as long as $\mathcal{E}_u$ does not happen along the path to a leaf, that leaf will end up with (at least) one ball. The next corollary notes that if $\mathcal{E}_u$ did happen along the path, then many balls have still reached~$u$---at least one for each leaf underneath~$u$. Indeed, the event $\mathcal{E}_u$ suggests that the balls are divided unequally between $u$'s children, so one of them might get no balls at all, but at $u$ itself, there must be enough balls to allow Lemma~\ref{lem:allgood} to hold and assign at least one ball to each leaf underneath~$u$.

\begin{corollary}
    \label{cor:atLeastOne}
    Let $u_0,u_1,\ldots,u_b$ be a path from $u_0\in N_0$ to some leaf~$u_b$, and assume $\ell_0\ge\log n$. Let $u_k$ be the first node along this path for which $\mathcal{E}_{u_k}$ holds. 
    Then, $\ell_{u_k}\ge 2^{b-k}$.
\end{corollary}
\begin{proof}
    Assume towards contradiction that $\ell_{u_k} < 2^{b-k}$. From $u_k$ downwards, re-assign the balls at the subtree rooted at $u_k$ so that exactly half (up to 1) balls reach each child of every internal node. Note that with this new arrangement, the event $\mathcal{E}_{u}$ does \emph{not} occur anymore at~$u_k$, nor in any $u$ decendant of~$u_k$. 
    We can now apply Lemma~\ref{lem:allgood} for each and every leaf~$z$ in the subtree rooted at~$u_k$ and learn that $\ell_z\ge 1$. 
    Thus, if each leaf in the subtree of $u_k$ has at least one ball, $u_k$ had at least as many balls as leaves below it, 
    $\ell_{u_k} = \sum_{
        z \in N_b \text{ descendant of }u_k 
         } 
      \ell_{z} \ge 2^{b-k}$. 
\end{proof}

We now turn to bound the cost of a node in which the event~$\mathcal{E}_u$ occurs. As alluded above, this cost may depend on ``neighbouring'' nodes of $u$. We need the following definition in order to properly defines these ``neighbors''.

\begin{definition}[Associated bins]\label{def:associated}
Define the function $f$, mapping  bins to nodes of the tree of deferred randomness in the following manner.
For every bin~$b_i$, consider the path from the root to~$b_i$. For the first node $u_j$ on that path where $\mathcal{E}_{u_j}$ holds, set $f(b_i)=u_j$. If there is no such node~$u_j$, simply set $f(b_i)=u_b$, where $u_b$ is the leaf associated with~$b_i$.

Now consider an arbitrary node $v$ for which $\mathcal{E}_{v}$ occurred. 
Let $B(v)=\{ b_{j}, b_{j+1}, \dots \}$ be the set of all associated bins, i.e., the set of bins $b_i$ that satisfy $f(b_i)=v$. 
In particular, if $f(b_i)=u_b$, where $u_b$ is the leaf associated with~$b_i$, then $B(u_b)=\{b_i\}$. 
\end{definition}

Let us also define the set of nodes where our process ``ends'': 
For any node~$u$, let $\END(u)$ denote the set of nodes $v$ such that (1) $v$ is a descendant of~$u$, and either (2a)  $\mathcal{E}_v$ occurs, or (2b) $v$ is a leaf, and $\mathcal{E}_{v'}$ does not occur for any $v'$ ancestor of~$v$. 
Recall that by definition, once $\mathcal{E}_v$ occurs, we stop the balls-infiltration process at this point, so $\mathcal{E}_{v'}$ never happens for any descendant of~$v$. That is, the set $\END(u)$ denotes exactly the nodes~$v$ in the subtree of~$u$ where the process stops, either due to $\mathcal{E}_v$, or when we reached a leaf.

We now continue to bounding the total normalized cost.
\begin{lemma}\label{lem:amortize}
Assuming no bad event happened before $N_0$, we can bound the cost of the root of the tree of deferred randomness as follows
\[
\COST(\Root)
\leq \sum_{v\in\END(\Root)} 8\times
\begin{cases}1 & v\text{ is a leaf} \\  \ell_v^2 & \text{otherwise}\end{cases}
\]
\end{lemma}

\begin{proof}
Consider some $v\in \END(\Root)$.
If $v$ is a leaf, we know by Lemma~\ref{lem:allgood}, that the single bin associated with this leaf, $B(v)=\{b_i\}$, has at least one ball in it.
Otherwise, $\mathcal{E}_{v}$ occurs and in this case recall we assume
nothing about the distribution of the~$\ell_v\ge 1$ balls (Corollary~\ref{cor:atLeastOne}) to the associated  bins $B(v)=\{b_j,\ldots,b_{j+|B(v)|-1}\}$. 
In either case, the normalized cost for sorting the elements in the buckets that correspond to~$B(v)$ depends on the locations of the closest pivot to the right of~$B(v)$ and the closest pivot to the left of~$B(v)$.
To this end consider $v_{\Left} = f(b_{j-1})$ and $v_{\Right}= f(b_{j+|B(v)|})$.
We have $B(v_{\Left})= \{ b_i \mid f(b_i)=v_{\Left}\}$ and $B(v_{\Right})=\{ b_i \mid f(b_i)=v_{\Right}\}$.
To deal with the edge cases,
if $j=1$ or $j+|B(v)|-1=\sqrt{n}/c$ we define $B(v_{\Left})=\emptyset$ or $B(v_{\Right})=\emptyset$, respectively.

\begin{tcolorbox}[boxrule=1pt,leftrule=3mm]
To illustrate the above notion, refer to  Figure~\ref{fig:example}. Say we wish to express the cost for sorting the bins associated with at the leaf~$u_{12}$. Note that  $\mathcal{E}_{u_{12}}$ occurs and $u_{12}\in \END(u_1)$. 
We have $B(u_{12})=\{b_9,b_{10}\}$.
Further, $u_{\Left}=u_5$  and $u_{\Right}=b_{11}$. 
We know that there is at least one ball in $B(u_{12})=\{b_9,b_{10}\}$, one in 
$B(b_{11})=\{b_{11}\}$ and at least one in $B(u_5)=\{b_5,b_6,b_7,b_8\}$. 
\end{tcolorbox}

As argued above,
there is at least one ball in the each of $B(v), B(v_{\Left}), B(v_{\Right})$. For the edge case where $B(v_{\Left})=\emptyset$, or $B(v_{\Right})$, technically there is no ball; note, however, that balls represent the borders of the buckets of elements to be compared, so we  can safely assume that the bucket ends at the last (or first) element of the rightmost (or leftmost) bin of~$B(v)$, in this case.

If so, the normalized cost for sorting all the  buckets that include elements from the bins~$B(v)$ is upper bounded by 
\[
(|B(v_{\Left})|+|B(v)|)^2 + (|B(v)|+|B(u_{\Right})|)^2 \leq 2|B(v_{\Left})|^2 + 4 |B(v)|^2 +  2|B(v_{\Right})|^2\text{.}
\]
That is, to maximize the cost we assume the very unlikely event that all the balls in $B(v_{\Left})$ reside in its left-most bin and the balls of $B(v_{\Right})$ reside in its right-most bin. There is also (at least one) ball in $B(v)$, so we get at least two buckets of sizes at most
$|B(v_{\Left})|+|B(v)|$ and
$|B(v)|+|B(v_{\Right})|$. Any other arrangement of balls can only lead to a lower cost.

While the cost for $B(v)$ depends on its neighbors, the above suggest that we can ``amortize'' the cost-per-node to be at most~$8\ell_v^2$ on average. 
That is, if we sum over \emph{all}
the nodes of the set $\END$,  
the additional terms for the neighboring nodes will balance out.

Formally, 
\begin{align*}
    \COST(\Root) &\le  
    \sum_{v\in\END(\Root)} 
    2|B(v_{\Left})|^2 + 4 |B(v)|^2 +  2|B(v_{\Right})|^2 \\
    &\le     
     \sum_{v\in\END(\Root)} 8|B(v)|^2 \text{,}
\intertext{This transition follows since each term $|B(v)|^2$ appears three times in the summation: twice with a factor of~2, when summing over $v_\Left$ and $v_\Right$, and once with a factor of~$4$, when summing over~$v$.
Note that, for any $v\in \END(\Root)$,  $v_\Left$, and $v_{\Right}$ are also in $\END(\Root)$, by their definition.}
    &\le  \sum_{\substack{v\in\END(\Root) \\ v\text{ is not a leaf}}}\!\!\!\!\!8\ell_v^2 \ +  \sum_{\substack{v\in\END(\Root) \\ v\text{ is a leaf}}}\!\!\!\!\!8
    \text{,}
\end{align*}
The last transition holds by applying Corollary~\ref{cor:atLeastOne} to all non-leaves, and recalling that a leaf contains a single bin.
\end{proof}

Next, we bound the probability that the bad event $\mathcal{E}_u$ happens at the node~$u$, as a function of the number of balls arriving to it.
\begin{lemma}
\label{lem:boundOnEu} 
Fix a node~$u$ 
and assume that $\ell_u$ balls are in~$u$. Then,
\[
\Pr{\mathcal{E}_u \mid \ell_u = \ell} \le 1/\ell^2\text{.}
\]
\end{lemma}
\begin{proof} 
Assume $u$ is assigned with $\ell_u=\ell$ balls and that $\mathcal{E}_{u}$ holds. This means that one of $u$'s children, say~$u_0$, receives less than $\ell/2 - g(\ell)$ balls, by definition.
It follows that the other children, $u_1$, must have received more than $\ell/2 + g(\ell)$ balls.

We can now bound the probability that $u_1$ receives that many balls, when balls are split uniformly at random among children.
We use a Chernoff bound with parameters $\mu=\ell/2$, $\delta=g(\ell)/\mu$ on the $\ell$ independent random variables indicators for the event that the $i$-th ball reaches~$u_1$, for $i=1,\ldots, \ell$. 
\begin{align*}
\Pr{\mathcal{E}_{u} \mid \ell_u=\ell  } \leq 
\exp\left(- \frac{\delta^2 \mu}{2}\right)
= \exp\left(- \frac{g^2(\ell)}{2\mu}\right)
<  %
1/\ell^2,
\end{align*}
where we used that $g(\ell) \triangleq \sqrt{2 \ell \log \ell}=\sqrt{4 \mu \log \ell}$.
\end{proof}
The bound in Lemma~\ref{lem:boundOnEu} is crucial: 
The impact of the occurrence of the event~$\mathcal{E}_u$ is a normalized cost of~$O(\ell_u^2)$. 
The Lemma shows that this high cost happens with probability~$1/\ell_u^2$, which mitigates the total cost once averaging over all the nodes.
All that is remaining in order to bound the total cost with high probability, 
is to show that the number of bad events that occur throughout the process is not too large.
To this end, we will use  
Hoeffding's inequality, to bound 
the normalized cost imposed by each layer of the tree, with high probability in~$n$. 
\begin{theorem}[Hoeffding~\cite{hoeffding63}]
\label{thm:hoe}
Let $(X_i)_{i=1}^n$ be independent random variables satisfying $a_i\le X_i\le b_i$, for $1\leq i \leq n$. Consider the sum $X=\sum_i X_i$. We have
$\Pr{X\geq \E{X} + \lambda} \leq \exp\left( \frac{-2\lambda^2}{\sum_i (b_i-a_i)^2} \right)$.
\end{theorem}

In order to complete the proof of the main theorem of this section, we analyze a slightly different version of the ball-in-bins process, described as follows.
Whenever $\mathcal{E}_u$ holds at some node~$u$, we still assume a normalized cost of~$O(\ell_u^2)$ for that node, but then, instead of stopping the process, we allow the  adversary to reassign the balls to the children of~$u$ in an arbitrary, under the restriction that $\mathcal{E}_u$ does not hold anymore. The adversary could cause $\mathcal{E}_v$ in the subtree of~$u$, and thus the cost of this non--stopping process can only be larger than the cost of the original process, that stops.
We are now ready to complete the proof of the main theorem.

\begin{proof}[Proof of Theorem~\ref{thm:main-rand-informal}]
We again consider the lower levels of the tree of deferred randomness, and apply  Hoeffding's inequality bound cumulatively on all the nodes on each level~$N_j$ for $j\in J$ where $J=\{0,\ldots, b-1$\}, i.e., all inner levels; recall that $b=\log(8\log(n)/c)$.
For $j \in \{0,\ldots, b\}$, let $V_j = \{ v \mid v \text{ in level } N_j\}$, be the set of nodes of level~$N_j$. 
We have that 
$|V_j| = \frac{\sqrt{n}}{8\log n} 2^{j}$ since there are 
$2^{T-b} = 2^{\log(\sqrt{n}/c)-\log(8\log(n)/c)}= \frac{\sqrt{n}}{c} \cdot \frac{c}{8\log n}
= \frac{\sqrt{n}}{8\log n}$  nodes in layer~$N_0$,  
the tree is binary, and we only look at the lower $b$~levels of the tree.

By Lemma~\ref{lem:amortize},
\begin{align}\label{eq:kettle}
\COST(root) 
\le \sum_{\substack{v\in\END(\Root) \\ v\text{ is not a leaf}}}\!\!\!\!\!8\ell_v^2 \ +  \sum_{\substack{v\in\END(\Root) \\ v\text{ is a leaf}}}\!\!\!\!\!8 
\ \le \ 
 8 \sum_{j \in J}\ \sum_{\substack{v\in\END(\Root) \\ \wedge v\in  V_j}}\ell_v^2 + 8|V_b|,
\end{align}
so we can bound the cost by analyzing each inner layer~$j\in J$ separately.
For any node~$v$, we define the cost of~$v$ by
\[
X_v =\begin{cases}
\ell_v^2 & \text{if $\mathcal{E}_v$ occurs}  \\  %
0 &  \text{otherwise}
\end{cases}. 
\]
Note that the value of $X_v$ dominates the real cost, since in this variant of the process that does not stop, we might have additional events $\mathcal{E}_v$ in nodes below nodes that belong to~$\END(v)$.
Then, the \emph{total} normalized cost imposed by the nodes in level~$j$ is given by $X^{(j)}=\sum_{v\in V_j}X_v$.

Since balls never get lost, the total number of balls, in any specific layer, is fixed,~$c\sqrt{n}$.
Further, recall that this analysis is conditioned on the high-probability event, proven at the beginning of this section, that each node in~$N_0$ begins the process with a number of balls in the range~$[\log n, 15\log n]$. 
The same upper bound holds also for any  node~$v$ in a level below~$N_0$, thus $\ell_v \le 15\log n$ balls, implying $0\le X_v \leq (15\log n)^2$. 

We can now compute the expectation of these variables, via 
Lemma~\ref{lem:boundOnEu}. We can assume $\mathcal{E}_v$ did not happen in levels above the current level~$j$ we bound (due to our altered process), thus Lemma~\ref{lem:allgood} suggests every node obtains at least one ball, $\ell_v\ge1$. Then,
\begin{align*}
\E{X_v} &\le 0\cdot \Pr{\lnot \mathcal{E}_v} + \max_\ell \{\E{X_v \mid \mathcal{E}_v , \ell_v=\ell}\cdot  \Pr{\mathcal{E}_v \mid \ell_v=\ell} \mid \ell \in [1, 15\log n] \} \\
&< 0+  \max_\ell\{\ell^2 \cdot (1/\ell)^2 \mid \ell \in [1, 15\log n] \}
\leq 1\text{,}
\end{align*}
and thus, for any $j\in J$, by linearity of expectations, $\E{X^{(j)}} \le |V_j| $.
We can now use Hoeffding's inequality (Theorem~\ref{thm:hoe}) on the $X_v$'s of some fixed level~$j$. Set $\lambda = |V_j|$, 
\begin{align*}
\Pr{X^{(j)}\geq 2 |V_j|} 
    &\leq \Pr{X^{(j)}\geq \E{X^{(j)}}+\lambda} 
     \le \exp\left( \frac{-2\lambda^2}{\sum_{v\in V_j} ((15\log n)^2-0)^2} \right)    \notag  \\
    &\le \exp\left( \frac{-2|V_j|^2}{|V_j| (15\log n)^4} \right)
    < \frac{1}{n^2}.
 \end{align*}
where the last transition holds for a sufficiently large~$n$, recalling that $|V_j|\ge \frac{\sqrt{n}}{8\log n}$. 
Finally, we can take a union bound on the $|J|=\log(8\log(n)/c)-1 < n$ lower levels of the tree.
Recall Eq.~\eqref{eq:kettle}; we get that, with high probability,
\begin{align*}
\COST(\Root) 
&\le
 8|V_b| + 8 \sum_{j \in J}\ \sum_{\substack{v\in\END(\Root) \\ \wedge v\in  V_j}}\ell_v^2 
\\
& \leq 8|V_b| + 8\sum_{j \in J }\sum_{v\in V_j} X_v =  8|V_b| + 8\sum_{j \in J }X^{(j)}\\
&\leq 8|V_b| + 8\sum_{j \in J } 2|V_j| = O(\sqrt{n}).
\end{align*}
The above \emph{normalized} cost requires multiplying with a factor of~$4c_1^2 n/t^2$ to get the real cost (measured in the number of comparators). This yields the desired claim and completes the proof.
\end{proof}

\section{Forming a sorted array out of comparators output}
\label{sec:FromCompsToSorting}
One crucial step is missing in converting the findings presented in this work into a ``standard'' sorting algorithm, where one obtains a sorted array of the element: combining the results of the separate $t$-comparators into a single total ordering of the $n$ elements. 
Clearly, once  we know the relative order of any two elements, we have all the information we need in order to generate a sorted array of all the element. 
For completeness, in this section we provide a brief description of how to complete the task of sorting out of the partial outcomes.

We assume that the input is given as the set $A=\{a_1,\ldots,a_n\}$,
and assume we have arrays 
$\RANK[1,...,n]$ and $\OUT[1,...n]$, all initialized to~$0$.

Consider the case of optimal deterministic sorting in~$d=1$ rounds, where each two elements get compared exactly by a single comparator (e.g., via Theorem~\ref{thm:minimal-sorting-params} or Lemma~\ref{lem:composition}).
Each $t$-comparator is given a set of elements say, $a_{i_1},\ldots,a_{i_t}$ and outputs their total ordering, say, it outputs a permutation of the elements (i.e., their indices) $j_1,\ldots,j_t$ such that (i) $\{j_1,\ldots,j_t\}=\{i_1,\ldots i_t\}$ and (2) $\val(a_{j_1}) \le \cdots \le \val(a_{j_t})$. For each $j_k \in \{j_1,\ldots,j_t\}$, update 
$\RANK[j_k] \mathrel{+\!=} k-1$, 
that is, increase the ranking of $a_{j_k}$ by the number of element smaller than it in this comparison.
Once we update the same for all the comparators, 
$\RANK[i]$ holds the rank of $a_i$ within the $n$ elements, for all $a_0,\ldots, a_n$.
Indeed, each element is compered exactly once with each other element, so $\RANK[i]$ holds exactly the number of elements smaller than $a_i$ overall.
Obtaining the output (sorted array) can be done in one additional pass on $\RANK[i]$, namely, 
for $i=1,\ldots, n$, write the element $a_i$ at index $\RANK[i]$ of the output array.
That is, 
$\OUT[\RANK[i]] \gets a_i$.
One can employ a similar idea in order to perform distributed sorting, where each $t$-comparator is a separate computer, and $\OUT,\RANK$ held as shared memory, with the cost of serializing the different writes per cell of the different machines. In our optimal deterministic sorting (Theorem~\ref{thm:main}) for $n=t^2$, we use $t^2+t$ different $t$-comparators, where each element is compared by exactly $t$ such comparators. This means that each cell in the shared memory is written by exactly $t$ machines.

\medskip
Employing a similar approach on the randomized sorting (Algorithm~\ref{alg:rand-general}) is partially possible. On the one hand, the buckets of the second round contain disjoint elements and can be sorted separately and then merged to one array in a trivial way. 
On the other hand, the sorting of Lemma~\ref{lem:loose-upperbound} might compare the same pair of elements several times, thus the above approach might not fit. 
In the general case, one can use some sort of mergesort, namely, go over the $t$-comparators one by one (or better, $t$ instances at a time), and merge each new $t$ sorted element to the sorted array of all previous comparators. Then, one can eliminate elements that appear multiple times.

\section{Conclusions and Future Directions}
\label{sec:conclusion}
In this work we studied the fundamental task of sorting $n$ elements with $t$-comparators, where the sorting algorithm is limited to a small number of interactive rounds. This setting, while interesting on its own, fits in particular to distributed and parallel settings where interactive communication is very costly while computation resources are moderately costly.

We dealt with both deterministic and randomized algorithms. In the deterministic case, we established connections between optimal sorting algorithms in one round and combinatorial design theory. While this connection allows optimal sorting for certain value of $n,t$, it also suggests the impossibility for other values (e.g., $t=6$). 
The question of the values of $n,t$ for which optimal sorting exists is isomorphic to the long-standing combinatorial question of deciding the values of $n,t$ for which the Steiner system $S(2, t, n)$ exists. 
We hope that an algorithmic approach could shed more light on this open question, e.g.\@ through the construction of composition theorems similar to Lemma~\ref{lem:composition},  or through explicit constructions for special cases.

Another interesting question is how the optimal number of $t$-comparators scales with the number of rounds. This topic was thoroughly examined in the literature for~$t=2$, and we extend the discussion to larger values of~$t$. In the same vein, in the randomized setting, we design algorithms that use only two rounds but utilize the same asymptotic number of comparators as the optimal $O(\log_t n)$-round $t$-quicksort algorithm.

We believe our findings might be useful in other distributed settings. For instance, in the Massively Parallel Computation model (MPC), where each worker machine performs the actions of one $t$-comparator, and all machines act in parallel. 
While our algorithm for $d=1$ rounds requires a large number of machines (i.e., more than $n/t$), it might make sense to consider a larger amount of rounds and how it tradeoffs the number of machines in use. For instance, could a sublinear number of machines be sufficient for $d=O(1)$ rounds?

\section*{Acknowledgments}
Research supported in part by the United States-Israel Binational Science Foundation (BSF) through Grant No.\@ 2020277.
R. Gelles would like to thank 
Paderborn University and CISPA---Helmholtz Center for Information Security 
for hosting him while part of this research was done. 
The authors would also like to thank the anonymous reviewers for multiple helpful comments.

\appendix
\section*{APPENDIX}

\section{A minimal construction for \texorpdfstring{$n=t^2$}{n=t\texttwosuperior} and \texorpdfstring{$t$}{t} power of a prime}
\label{APP:minConstrustions}
In this part we focus on the special case of $n=t^2$. This case has been extensively studied in the literature (see, e.g., \cite{SS86,SSS86} and many followup work).

Our construction fits the case where $t$ is a prime number. In this case, simple shift-permutations of the elements cover the entire element-space without repeating any element twice. This result can be extended to any number $t$ which is a power of a prime, taking leverage on the finite field $GF(t)$. 
To ease the read, we provide here the simple and more intuitive construction of $t$ being a prime. 
Later on,
we show that this result extends to $t$ power of a prime in a straightforward manner.
\begin{theorem}
\label{thm:t-prime}
Suppose $t$ is prime and $n=t^2$. Then, sorting $n$ elements can be done using $t+t^2$ many $t$-comparators.
\end{theorem}
Note that this is optimal by Corollary~\ref{cor:triv-lowerBound} as $t+t^2 = \binom{t^2}{2}/\binom{t}{2}$.

\begin{proof}
Let us rename the $n$ elements
$A=\{a_0,\ldots, a_{n-1}\}$, to ease the notations.

    We first split the $n=t^2$ elements in $A$
    into disjoint subset $S_1,\ldots, S_t$ and compare the elements within each subset. This utilizes $t$ comparators.

    Let $M_0 \in A^{t \times t}$ be the $t\times t$ matrix whose $i$-th row is $S_i$ (in some arbitrary order; suppose the natural order by their index, i.e., $a_0, a_1, a_2 \ldots$, as portrayed in Figure~\ref{fig:M_i-example}). 
    We further construct $t-1$ additional matrices, $M_1,\ldots,M_{t-1}$, where each row in $M_i$ is constructed from $M_0$ by rotating its $j$-th row $0\le j\le t-1$ by $i\cdot j$ positions to the right. 
    We think of each column of $M_i$ as $t$ elements that will be given as an input to a $t$-comparator. Hence, each $M_i$ describes $t$ such comparators. Overall, these matrices describe  $t^2$ comparators. See Figure~\ref{fig:M_i-example}, for the first two such matrices.
    \begin{figure}[ht]
        \centering
        \caption{The matrices $M_0$, $M_1$}
        \label{fig:M_i-example}
        \begin{align*}
            M_0 =\begin{pmatrix}
                \mathbf{a_0} & a_1 & \cdots & a_{t-1} \\
                \mathbf{a_{t}} & a_{t+1} & & a_{2t-2} \\
                \vdots & & \ddots \\
                \mathbf{a_{t^2-t}} &   &  & a_{t^2-1}
            \end{pmatrix}
            \qquad
            M_1 =\begin{pmatrix}
                \mathbf{a_0} & a_1 & & \cdots & a_{t-1} \\
                a_{2t-1} & \mathbf{a_t}& a_{t+1} \\
                \vdots & & & \ddots \\
                a_{t^2-t+1} &   &  &a_{t^2-1} &\mathbf{a_{t^2-t}}
            \end{pmatrix}
        \end{align*}
    \end{figure}
    
    We now argue that any two elements are compared by the above construction.
    Let $a_i,a_j$ be arbitrary two distinct elements with $0\le i,j < t^2$. Write $i=k_i \cdot t + r_i$ and $j=k_j \cdot t + r_j$ with $k_i,k_j,r_i,r_j \in \{0,1,\ldots, t-1\}$.
    If $k_i=k_j$, then $a_i,a_j$ both belong in $S_{k_i}=S_{k_j}$ and thus are being compared by one of the first $t$ comparators.

    Otherwise, let $c = (r_i-r_j)(k_j-k_i)^{-1}$ over the field $\Z_{t}$. 
    We argue that $M_c$ has a column that holds both $a_i$ and $a_j$, and thus they are compared. Indeed, $a_j$ appears in the $(c \cdot k_i+r_i \mod t)$-th column in $M_c$ and $a_j$ is in the $(c \cdot k_j+r_j \mod t)$-th column. Substituting for~$c$, it is easy to verify that over~$\Z_t$, 
    \[
        (r_i-r_j)(k_j-k_i)^{-1} k_i + r_i =
        (r_i-r_j)(k_j-k_i)^{-1} k_j + r_j,
    \]
     thus $a_i$ and $a_j$ belong in the same column (comparator) of~$M_c$.
\end{proof}

\label{app:t-prime-power}

We can now extend Theorem~\ref{thm:t-prime} to the case where $t$ is a power of a prime. The proof follows by taking a similar approach to the proof of Theorem~\ref{thm:t-prime}, but considering the finite field~$GF(t)$, when $t$ is a prime-power, instead of the field~$\Z_t$, when $t$ is prime.
This yields the following.

\begin{theorem}
\label{thm:t-prime-power}
Suppose $t=p^m$ is a prime power and let $n=t^2$. Then, sorting $n$ elements can be done using $t+t^2$ many $t$-comparators.
\end{theorem}
\begin{proof}
     We first split the $n=t^2$ elements, which we denote here $A=\{a_0,\ldots, a_{n-1}\}$,
    into disjoint subset $S_1,\ldots, S_t$ (greedily, by increasing order of the elements) and compare the elements within each subset. This utilizes $t$ comparators.

    Let $M_0 \in A^{t \times t}$ be the $t\times t$ matrix whose $i$-th row is $S_i$ in increasing element order (see Figure~\ref{fig:M_0-example}). 
    \begin{figure}[ht]
        \centering
        \caption{The matrix $M_0$}
        \label{fig:M_0-example}
        \begin{align*}
            M_0 =\begin{pmatrix}
                \mathbf{a_0} & a_1 & \cdots & a_{t-1} \\
                \mathbf{a_{t}} & a_{t+1} & & a_{2t-2} \\
                \vdots & & \ddots \\
                \mathbf{a_{t^2-t}} &   &  & a_{t^2-1}
            \end{pmatrix}
        \end{align*}
    \end{figure}
    We further construct $t-1$ additional matrices, $M_1,\ldots,M_{t-1}$, where each row in $M_i$ is constructed from $M_0$ by ``rotating'' its $j$-th row $0\le j\le t-1$ by $i\cdot j$ steps. Here, ``rotating'' the row has the following meaning. Let $\varphi: \{0,1,\ldots, t-1\} \to GF(t)$ be an isomorphism between $\{0,1,\ldots, t-1\} $ and the additive group of the  finite field $GF(t)$; e.g., by concatenating the natural mapping between $\{0,1,\ldots, t-1\} $ and $Z^m_p$, i.e., $(\alpha_0,\ldots, \alpha_{m-1})\in Z^m_p\  \cong \ \sum_i \alpha_i p^i \in \{0,1,\ldots, t-1\} $, with a standard isomorphism between $Z^m_p$ and the additive  group of the  finite field $GF(p^m)$, which exists since $GF(p^m)$ can be seen as a vector space of dimension $m$ over $GF(p)=\Z_p$.  
    See for instance~\cite[Theorem~2 and Note~2]{Mittal15}

    Then, rotation by a single step of some $t$-ary vector (i.e., a row) means placing the element of index $i$ in index $\varphi^{-1} \bigl(\varphi(i)+\varphi(1)\bigr)$.
    More generally, rotation by $\alpha$ steps, where $\alpha\in GF(t)$ means moving the element in index~$i$ to index $\varphi^{-1} (\varphi(i)+\alpha)$. 
    Thus, back to the construction of $M_i$, we have that 
    the $j$-th row in $M_i$ is constructed from the $j$-th row of~$M_0$ by rotating it $\varphi(i) \cdot \varphi(j)$ steps, where the multiplication is over $GF(t)$. I.e., the element in row~$r$ and column~$c$ in~$M_0$ will be placed in row~$r$ and column $\varphi^{-1} \bigl( \varphi(c) + \varphi(r)\cdot\varphi(i)\bigr)$ in~$M_i$.

    \medskip
    
    We think of each column of $M_i$ as $t$ elements that will be given as an input to a $t$-comparator. Hence, each $M_i$ describes $t$ such comparators. Overall, these matrices describe  $t^2$ comparators.
    With $t$ comparators for comparing $S_1,\ldots,S_t$ and $t^2$ for the matrices, we end up with $t^2+t = \binom{t^2}{2}/\binom{t}{2}$ comperators, as stated.

    \smallskip

    We now argue that any two elements are compared by the above construction.
    Let $a_i,a_j$ be arbitrary two distinct elements with $0\le i,j < t^2$. Write $i=k_i \cdot t + r_i$ and $j=k_j \cdot t + r_j$ with $k_i,k_j,r_i,r_j \in \{0,1,\ldots, t-1\}$.
    If $k_i=k_j$, then $a_i,a_j$ both belong in $S_{k_i}=S_{k_j}$ and thus are being compared by one of the first $t$ comparators.

    Otherwise, let $c = \varphi^{-1} \bigl( \frac{\varphi(r_i)-\varphi(r_j)} {\varphi(k_j)-\varphi(k_i)}\bigl)$. 
    We argue that $M_c$ has a column that holds both $a_i$ and~$a_j$. %
    Indeed, $a_i$ resides in index $r_i$ in the $k_i$-th row of $M_0$, and similarly, $a_j$ resides in index~$r_j$ of the $k_j$-th row of~$M_0$. 
    By construction, $a_i$ appear in $M_c$ in column number  
    \[
     \varphi^{-1} \bigl(
        \varphi(c) \cdot \varphi(k_i) + \varphi(r_i)
        \bigr).
    \]
    Similarly, $a_j$ will appear in~$M_c$ in column number
    \[
     \varphi^{-1} \bigl(
        \varphi(c) \cdot \varphi(k_j) + \varphi(r_j)
        \bigr).
    \]
     Substituting for~$c$, it is easy to verify these two column numbers are in fact equal, 
     \begin{align*}
         \frac{\varphi(r_i)-\varphi(r_j)} {\varphi(k_j)-\varphi(k_i)} \varphi(k_j)+\varphi(r_j) = 
         \frac{\varphi(r_i)-\varphi(r_j)} {\varphi(k_j)-\varphi(k_i)}\varphi(k_i)+\varphi(r_i)
     \end{align*}
     This equality implies that $a_i$ and $a_j$ belong in the same column (i.e., the same comparator) of~$M_c$ and thus will be compared by our construction.
\end{proof}

Theorems~\ref{thm:t-prime} and~\ref{thm:t-prime-power} prove Theorem~\ref{thm:minimal-sorting-params} part (1).
In hindsight, after establishing the connection between sorting and Steiner system and other equivalent objects such as affine/projective planes and Latin Squares, one can easily see that the constructions in the proofs of Theorems~\ref{thm:t-prime} and~\ref{thm:t-prime-power} are equivalent to the (orthogonal Latin square) constructions of Bose~\cite{Bose38} and Stevens~\cite{Stevens39}.

\section{Simulations: Our algorithm and the state-of-the-art algorithm}
\label{app:simulations}
Let us compare our Algorithm~\ref{alg:rand-general-intro}  to the state-of-the-art  quicksort algorithm with $t$-comparators, developed by Beigel and Gill~\cite{beigel1990sorting}. 
Their algorithm  works essentially as follows: randomly select $t/\log t$ pivot elements and use them to split all the elements into disjoint subsets. Now, recursively sort any subset of size exceeding~$t$.

\begin{figure}[ht]
\begin{subfigure}{.5\textwidth}
  \centering
  \includegraphics[width=.9\linewidth]{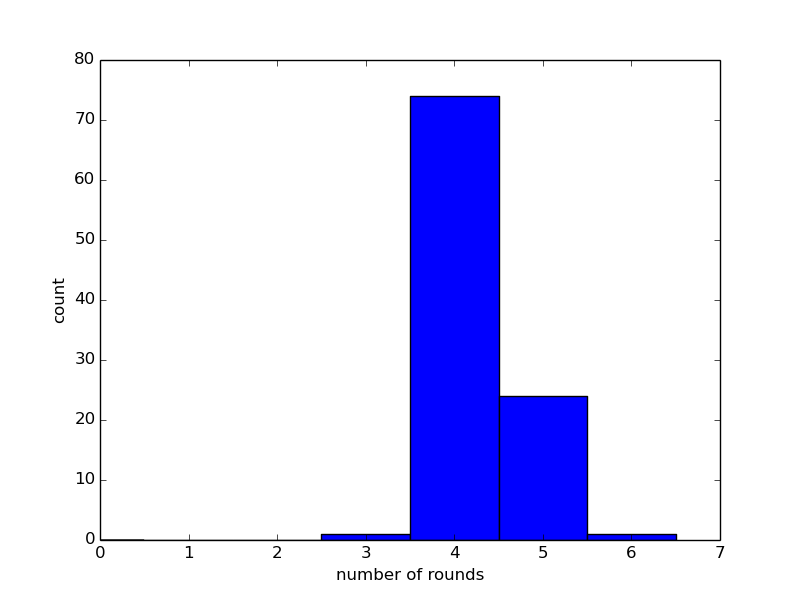}
  \caption{$t=10$, $n=100$}
  \label{fig:sfig1}
\end{subfigure}%
\begin{subfigure}{.5\textwidth}
  \centering
  \includegraphics[width=.9\linewidth]{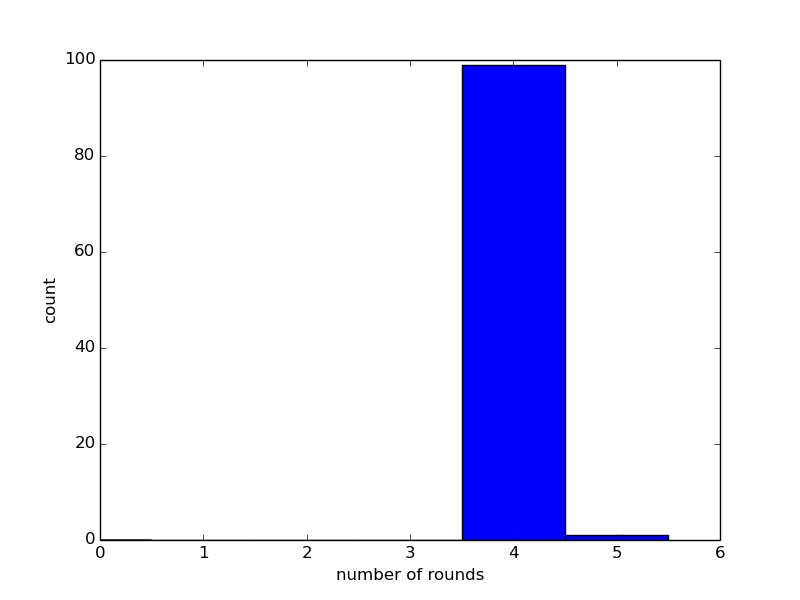}
  \caption{$t=100$, $n=10000$}
  \label{fig:sfig2}
\end{subfigure}
\caption{A histogram of the number of rounds required to the completion of the algorithm in~\cite{beigel1990sorting} for the case of $n=t^2$ with (a) $t=10$ and (b) $t=100$. 
Each histogram is based on $100$ repeated independent  instances. In both $t$ values, the average number of rounds is above~$4$. }
\label{fig:fig}
\end{figure}

The analysis in~\cite{beigel1990sorting} proves that the number of $t$-comparators utilized throughout this algorithm is 
$\frac{n \log n}{t \log t}(1+o(1))$, which is asymptotically optimal. 
The same analysis suggests the algorithm takes $\log_{m/2}(n)$ rounds, where $m=t/(2\log t\ln t)$. (The basis of the $\log$ in~$m$ is not defined in~\cite{beigel1990sorting} and we take it to base~$e$.)
It is easy to verify that this function 
approaches $\frac{\log n}{\log t}$ rounds, for sufficiently large~$t$. In particular, for $n=t^c$, the function approaches $c$ rounds as~$t\to \infty$.
We would like to compare this to our algorithm, that guarantees~$d=2$ rounds, regardless of~$t$.

To be concrete, let us consider  the case of $n=t^2$.
In this case, $\log_{m/2}(t^2)$ tends asymptotically to~2 when $t\to\infty$.
To demonstrate the behavior of the recursive algorithm we have performed Monte-Carlo simulations that measure the number of rounds it takes to sort $n=t^2$ elements, with $t=10$ and $t=100$. 
The results are depicted in Figure~\ref{fig:fig}. 
Our findings indicate that, for these values of~$t$, the average number of rounds for $n=t^2$ is not~$2$, but rather~$4$.


\begin{thebibliography}{10}

\bibitem{AKS83}
M.~Ajtai, J.~Koml\'{o}s, and E.~Szemer\'{e}di.
\newblock An ${O(n \log n)}$ sorting network.
\newblock In {\em Proceedings of the Fifteenth Annual ACM Symposium on Theory
  of Computing}, STOC '83, page 1–9, 1983.

\bibitem{akl1985}
S.~G. Akl.
\newblock {\em Parallel sorting algorithms}, volume~12.
\newblock Academic press, 1985.

\bibitem{AA87}
N.~Alon and Y.~Azar.
\newblock The average complexity of deterministic and randomized parallel
  comparison sorting algorithms.
\newblock In {\em 28th Annual Symposium on Foundations of Computer Science
  ({SFCS} 1987)}, pages 489--498, 1987.

\bibitem{AA88}
N.~Alon and Y.~Azar.
\newblock Sorting, approximate sorting, and searching in rounds.
\newblock {\em SIAM Journal on Discrete Mathematics}, 1(3):269--280, 1988.

\bibitem{AAV86}
N.~Alon, Y.~Azar, and U.~Vishkin.
\newblock Tight complexity bounds for parallel comparison sorting.
\newblock In {\em 27th Annual Symposium on Foundations of Computer Science
  ({SFCS} 1986)}, pages 502--510, 1986.

\bibitem{Batcher68}
K.~E. Batcher.
\newblock Sorting networks and their applications.
\newblock In {\em Proceedings of the April 30--May 2, 1968, Spring Joint
  Computer Conference}, AFIPS '68 (Spring), page 307–314, 1968.

\bibitem{beigel1990sorting}
R.~Beigel and J.~Gill.
\newblock Sorting n objects with a k-sorter.
\newblock {\em IEEE Transactions on Computers}, 39(5):714--716, 1990.

\bibitem{Bose38}
R.~C. Bose.
\newblock On the application of the properties of galois fields to the problem
  of construction of hyper-gr{\ae}co-latin squares.
\newblock {\em Sankhy{\=a}: The Indian Journal of Statistics (1933-1960)},
  3(4):323--338, 1938.

\bibitem{BMP19}
M.~Braverman, J.~Mao, and Y.~Peres.
\newblock Sorted top-k in rounds.
\newblock In A.~Beygelzimer and D.~Hsu, editors, {\em Proceedings of the
  Thirty-Second Conference on Learning Theory}, volume~99 of {\em PMLR}, pages
  342--382, 25--28 Jun 2019.

\bibitem{BMW16}
M.~Braverman, J.~Mao, and S.~M. Weinberg.
\newblock Parallel algorithms for select and partition with noisy comparisons.
\newblock In {\em Proceedings of the Forty-Eighth Annual ACM Symposium on
  Theory of Computing}, STOC '16, page 851–862, 2016.

\bibitem{BR49}
R.~H. Bruck and H.~J. Ryser.
\newblock The nonexistence of certain finite projective planes.
\newblock {\em Canadian Journal of Mathematics}, 1(1):88--93, 1949.

\bibitem{Chiang01}
Y.~Chiang.
\newblock {\em Sorting networks using k-comparators}.
\newblock PhD thesis, University of Cape Town, 2001.

\bibitem{CLRS}
T.~H. Cormen, C.~E. Leiserson, R.~L. Rivest, and C.~Stein.
\newblock {\em Introduction to algorithms}.
\newblock MIT press, 4th edition, 2022.

\bibitem{CS92}
R.~Cypher and J.~L. Sanz.
\newblock Cubesort: A parallel algorithm for sorting n data items with
  s-sorters.
\newblock {\em Journal of Algorithms}, 13(2):211--234, 1992.

\bibitem{DKP23}
N.~Dobrokhotova-Maikova, A.~Kozachinskiy, and V.~Podolskii.
\newblock {Constant-Depth Sorting Networks}.
\newblock In Y.~Tauman~Kalai, editor, {\em 14th Innovations in Theoretical
  Computer Science Conference (ITCS 2023)}, volume 251 of {\em LIPIcs}, pages
  43:1--43:19. Schloss Dagstuhl -- Leibniz-Zentrum f{\"u}r Informatik, 2023.

\bibitem{10.1145/146370.146381}
V.~Estivill-Castro and D.~Wood.
\newblock A survey of adaptive sorting algorithms.
\newblock {\em ACM Comput. Surv.}, 24(4):441–476, dec 1992.

\bibitem{euler1782}
L.~Euler.
\newblock Recherches sur un nouvelle esp{\'e}ce de quarr{\'e}s magiques.
\newblock {\em Verhandelingen uitgegeven door het zeeuwsch Genootschap der
  Wetenschappen te Vlissingen}, pages 85--239, 1782.

\bibitem{GG94}
M.~Grannell and T.~Griggs.
\newblock An introduction to steiner systems.
\newblock {\em Mathematical Spectrum}, 26(3):74--80, 1994.

\bibitem{HH81}
R.~H\"{a}ggkvist and P.~Hell.
\newblock Parallel sorting with constant time for comparisons.
\newblock {\em SIAM Journal on Computing}, 10(3):465--472, 1981.

\bibitem{quicksort}
C.~A.~R. Hoare.
\newblock {Quicksort}.
\newblock {\em The Computer Journal}, 5(1):10--16, 01 1962.

\bibitem{hoeffding63}
W.~Hoeffding.
\newblock Probability inequalities for sums of bounded random variables.
\newblock {\em Journal of the American Statistical Association},
  58(301):13--30, 1963.

\bibitem{Hughes_Piper_1985}
D.~R. Hughes and F.~Piper.
\newblock {\em Design Theory}.
\newblock Cambridge University Press, 1985.

\bibitem{SSS86}
S.~Isaacd, S.~Sandeep, and A.~Shamir.
\newblock Shear sort-a true two-dimensional sorting technique for vlsi
  networks.
\newblock In {\em International Conference on Parallel Processing}, pages
  903--908, 1986.

\bibitem{KK92}
C.~Kaklamanis and D.~Krizanc.
\newblock Optimal sorting on mesh-connected processor arrays.
\newblock In {\em Proceedings of the Fourth Annual ACM Symposium on Parallel
  Algorithms and Architectures}, SPAA '92, page 50–59, 1992.

\bibitem{knuthVol3-Sorting}
D.~E. Knuth.
\newblock {\em Art of computer programming, volume 3: Sorting and Searching}.
\newblock Addison-Wesley Professional, 2nd edition, April 1998.

\bibitem{leighton85}
T.~Leighton.
\newblock Tight bounds on the complexity of parallel sorting.
\newblock {\em IEEE Transactions on Computers}, C-34(4):344--354, 1985.

\bibitem{LW11}
C.~Lenzen and R.~Wattenhofer.
\newblock Tight bounds for parallel randomized load balancing: extended
  abstract.
\newblock In {\em Proceedings of the Forty-Third Annual ACM Symposium on Theory
  of Computing}, STOC '11, page 11–20, 2011.

\bibitem{10.1145/356593.356594}
W.~A. Martin.
\newblock Sorting.
\newblock {\em ACM Comput. Surv.}, 3(4):147–174, dec 1971.

\bibitem{Mittal15}
R.~Mittal.
\newblock Lecture notes, introduction to abstract algebra, {IIT Kanpur}, 2015.

\bibitem{OZ96}
S.~Olarin and S.~Zheng.
\newblock Sorting n items using a p-sorter in optimal time.
\newblock In {\em Proceedings of SPDP '96: 8th IEEE Symposium on Parallel and
  Distributed Processing}, pages 264--272, 1996.

\bibitem{PP89}
B.~Parker and I.~Parberry.
\newblock Constructing sorting networks from k-sorters.
\newblock {\em Information Processing Letters}, 33(3):157--162, 1989.

\bibitem{Pascoe18}
A.~Pascoe.
\newblock Affine and projective planes, 2018.
\newblock MSU Graduate Theses. 3233.
  \url{https://bearworks.missouristate.edu/theses/3233}.

\bibitem{PT11}
B.~Patt-Shamir and M.~Teplitsky.
\newblock The round complexity of distributed sorting: extended abstract.
\newblock In {\em Proceedings of the 30th Annual ACM SIGACT-SIGOPS Symposium on
  Principles of Distributed Computing}, PODC '11, page 249–256, 2011.

\bibitem{reid2012steiner}
C.~Reid and A.~Rosa.
\newblock Steiner systems $ s (2, 4, v) $-a survey.
\newblock {\em The Electronic Journal of Combinatorics}, pages DS18--Feb, 2012.

\bibitem{RSS85}
D.~Rotem, N.~Santoro, and J.~B. Sidney.
\newblock Distributed sorting.
\newblock {\em IEEE Transactions on Computers}, C-34(4):372--376, 1985.

\bibitem{SS86}
C.-P. Schnorr and A.~Shamir.
\newblock An optimal sorting algorithm for mesh connected computers.
\newblock In {\em Proceedings of the eighteenth annual ACM symposium on Theory
  of computing}, pages 255--263, 1986.

\bibitem{SZW14}
F.~Shi, Z.~Yan, and M.~Wagh.
\newblock An enhanced multiway sorting network based on n-sorters.
\newblock In {\em 2014 IEEE Global Conference on Signal and Information
  Processing (GlobalSIP)}, pages 60--64, 2014.

\bibitem{singh2018survey}
D.~P. Singh, I.~Joshi, and J.~Choudhary.
\newblock Survey of gpu based sorting algorithms.
\newblock {\em International Journal of Parallel Programming}, 46:1017--1034,
  2018.

\bibitem{Stevens39}
W.~Stevens.
\newblock The completely orthogonalized latin square.
\newblock {\em Annals of Eugenics}, 9(1):82--93, 1939.

\bibitem{Tarry1900}
G.~Tarry.
\newblock Le problème de 36 officiers.
\newblock {\em Compte Rendu de l'Association Française pour l'Avancement de
  Science Naturel}, 1900.
\newblock vol. 1 (1900), 122-123; vol. 2 (1901), 170-203.

\bibitem{TK77}
C.~D. Thompson and H.~T. Kung.
\newblock Sorting on a mesh-connected parallel computer.
\newblock {\em Commun. ACM}, 20(4):263–271, apr 1977.

\bibitem{WEGNER84}
L.~M. Wegner.
\newblock Sorting a distributed file in a network.
\newblock {\em Computer Networks (1976)}, 8(5):451--461, 1984.

\bibitem{WILSON75III}
R.~M. Wilson.
\newblock An existence theory for pairwise balanced designs, {III}: Proof of
  the existence conjectures.
\newblock {\em Journal of Combinatorial Theory, Series A}, 18(1):71--79, 1975.

\end{thebibliography}
\end{document}